\documentclass{article}

\usepackage{arxiv}
\usepackage[unicode]{hyperref}
\usepackage[utf8]{inputenc} 
\usepackage[T1]{fontenc}    
\usepackage{hyperref}       
\usepackage{url}            
\usepackage{booktabs}       
\usepackage{amsfonts}       
\usepackage{nicefrac}       
\usepackage{microtype}      
\usepackage{lipsum}		
\usepackage{doi}

\usepackage[pdftex]{graphicx}
\usepackage{float}

\usepackage[ruled, vlined, linesnumbered]{algorithm2e}
\usepackage{amsthm, nccmath}

\usepackage{amsmath, amsthm, enumerate,boxedminipage}
\usepackage{color}
\usepackage{comment}
\usepackage{hyperref}
\usepackage{booktabs}

\theoremstyle{definition}
\newtheorem*{def*}{Definition}

\newtheorem{observation}{Observation}
\newtheorem{theorem}{Theorem}[section]
\theoremstyle{definition}

\newtheorem{lemma}{Lemma}

\theoremstyle{definition}
\newtheorem*{prf*}{Proof}
  
\theoremstyle{definition}
\newtheorem*{prfthm*}{Proof of Theorem}

\newcommand{\ajay}[1]
{{\color{red}\underline{\textsf{Ajay:}}} {\color{red} \emph{#1}}}

\usepackage{array}
\newcolumntype{P}[1]{>{\centering\arraybackslash}p{#1}}
\newcolumntype{M}[1]{>{\centering\arraybackslash}m{#1}}
\usepackage{multirow}

\newcommand{\red}[1]{{\textcolor{red}{#1}}}

\newcommand{\A}{\mathcal{A}}

\newcommand{\sync}{$\mathcal {SYNC}$}
\newcommand{\async}{$\mathcal {ASYNC}$}

\newcommand{\ie}{{\em i.e.},~}

\newcommand{\aset}{A}

\newcommand{\amax}{a_{\mathrm{max}}}

\newcommand{\lctn}{h}
\newcommand{\id}{\mathtt{ID}}
\newcommand{\pin}{\mathtt{pin}}
\newcommand{\pout}{\mathtt{pout}}
\newcommand{\settled}{\mathtt{settled}}

\newcommand{\parent}{parent}

\newcommand{\checked}{\mathtt{checked}}
\newcommand{\nxt}{next}

\newcommand{\rootsync}{\mathbf{Rooted\-Sync\-Disp}}
\newcommand{\rootasync}{\mathbf{RootedAs\-ync\-Disp}}

\newcommand{\seeker}{\texttt{seeker }}

\begin{document}

\title{Dispersion is (Almost) Optimal under (A)synchrony}

\author{Ajay D. Kshemkalyani\thanks{University of Illinois Chicago, USA, {\em ajay@uic.edu}}, Manish Kumar\thanks{IIT Madras, India, {\em manishsky27@gmail.com}}, Anisur Rahaman Molla\thanks{Indian Statistical Institute, Kolkata, {\em anisurpm@gmail.com}}, and Gokarna Sharma\thanks{Kent State University, USA, {\em gsharma2@kent.edu}}}

\maketitle

\sloppy

\begin{abstract}
The dispersion problem has received much attention recently in the distributed computing literature.  
In this problem, $k\leq n$
 agents placed initially arbitrarily on the nodes of an $n$-node, $m$-edge anonymous graph of maximum degree $\Delta$ have to reposition autonomously to reach a configuration in which each agent is on a distinct node of the graph. Dispersion is interesting as well as important due to its connections to many fundamental coordination problems by mobile agents on graphs, such as exploration, scattering, load balancing, relocation of self-driven electric cars (robots) to recharge stations (nodes), etc.  The objective has been to provide a solution that optimizes simultaneously time and memory complexities. There exist graphs for which the lower bound on time complexity is $\Omega(k)$. Memory complexity is $\Omega(\log k)$ per agent independent of graph topology. The state-of-the-art algorithms have  (i) time complexity $O(k\log^2k)$ and memory complexity $O(\log(k+\Delta))$ under the synchronous setting [DISC'24] and (ii) time complexity $O(\min\{m,k\Delta\})$ and memory complexity $O(\log(k+\Delta))$ under the asynchronous setting [OPODIS'21].     
In this paper, we improve substantially on this state-of-the-art.  Under the synchronous setting as in [DISC'24], we present the first optimal $O(k)$ time algorithm keeping memory complexity $O(\log (k+\Delta))$. 
Under the asynchronous setting as in [OPODIS'21], we present the first algorithm with time complexity $O(k\log k)$ keeping memory complexity $O(\log (k+\Delta))$, which is time-optimal within an $O(\log k)$ factor despite asynchrony. Both results were obtained through novel techniques to quickly find empty nodes to settle agents, which may be of independent interest. 
\end{abstract}


\section{Introduction}\label{sec: introduciton}
We consider the dispersion problem which has been extensively studied in the distributed computing literature recently.
The {\em dispersion} problem asks $k\leq n$ mobile
 agents placed initially arbitrarily on the nodes of an $n$-node anonymous graph to reposition autonomously to reach a configuration in which each agent is on a distinct node of the graph. This problem is interesting as well as important due to its connections to many fundamental agent coordination problems, such as exploration, scattering, load balancing, relocating self-driven electric cars (agents) to recharge stations (nodes), etc.  The objective has been to provide a solution that optimizes simultaneously time and memory complexities. 
 Time complexity is measured as the time duration to achieve dispersion starting from any initial configuration. Memory complexity is measured as the number of bits stored in the persistent memory at each agent (graph nodes are memory-less and cannot store any information).
 There exist graphs for which any dispersion algorithm exhibits time complexity $\Omega(k)$. Suppose two nodes of a graph are at a distance of $\Omega(k)$, then to disperse from one node to the furthest node takes $\Omega(k)$ time (e.g., line graph). Memory complexity is  $\Omega(\log k)$ for any solution since at least one agent needs $\Omega(\log k)$ bits to store its unique ID.

 \begin{table*}[!t]
\centering
\footnotesize
\begin{tabular}
{P{2cm}P{2cm}P{5.0cm}P{3.5cm}}
\toprule

{\bf Model} &   {\bf Paper} & {\bf Time Complexity} & {\bf Memory Complexity}\\

\toprule

Any & - &  $\Omega(k)$ & $\Omega(\log k)$ \\
\toprule
&& {\bf Rooted Initial Configurations} & \\
\toprule
& \cite{sudo24} & $O(k)$&$O(\Delta + \log k)$\\
{\sync}   & \cite{sudo24} & $O(k \log k) 
 $&$O(\log (k+\Delta))$\\
&  Theorem~\ref{thm: root_sym} & $O(k)$ & $O(\log (k+\Delta))$\\
\toprule
&  \cite{KshemkalyaniMS22} & $O(D\Delta (D+\Delta))$&$O(D + \Delta \log k)$\\
{\async}  & \cite{KshemkalyaniS21-OPODIS} & $O(\min\{m,k\Delta\})$&$O(\log (k+\Delta))$\\
&  Theorem~\ref{theorem:rootal}  & $O(k \log k) 
 $&$O(\log (k+\Delta))$\\

\toprule
&& {\bf General Initial Configurations} & \\
\toprule
&  \cite{sudo24} & $O(k \log^2 k) 
 $& $O(\log (k+\Delta))$\\
{\sync} & Theorem~\ref{theorem:gensync} & $O(k)$& $O(\log (k+\Delta))$\\

\toprule
&  \cite{KshemkalyaniS21-OPODIS} & $O(\min\{m,k\Delta\})$&$O(\log (k+\Delta))$\\
{\async}  & Theorem~\ref{theorem:genasync} & $O(k \log k) 
 $&$O(\log (k+\Delta))$\\

\bottomrule
\end{tabular}
\caption{An illustration of the proposed dispersion results and comparison with the state-of-the-art.} 
\label{tbl: Comparative_analysis-rooted}
\vspace{-6mm}
\end{table*}

There are two state-of-the-art algorithms, one in the synchronous setting due to Sudo, Shibata, Nakamura, Kim, and Masuzawa \cite{sudo24} appeared recently in [DISC'24] and another in the asynchoronous setting due to Kshemkalyani and Sharma \cite{KshemkalyaniS21-OPODIS} appeared in [OPODIS'21]. In the {\em synchronous} setting ({\sync}), all agents 
perform their 
operations simultaneously in synchronized rounds (or steps), and hence time complexity (of the algorithm) is measured in rounds. However, in the {\em asynchronous} setting ({\async}), agents become active at arbitrary times and perform their operations in arbitrary duration, and hence time complexity is measured in epochs -- an {\em epoch} represents the time interval in which each agent becomes active at least once. In {\sync}, an epoch is a round. 
In particular,
Sudo, Shibata, Nakamura, Kim, and Masuzawa \cite{sudo24} [DISC'24] presented an algorithm with time complexity $O(k\log^2k)$ rounds and memory complexity $O(\log (k+\Delta))$ bits in {\sync}. 
Kshemkalyani and Sharma \cite{KshemkalyaniS21-OPODIS} [OPODIS'21] presented an algorithm with time complexity $O(\min\{m,k\Delta\})$ epochs and memory complexity $O(\log (k+\Delta))$ bits in {\async}. Here $m$ and $\Delta$ are the number of edges and the maximum degree of the graph respectively.

\noindent{\bf Contributions.}
In this paper, we improve substantially on this state-of-the-art through two results, one in {\sync} as in \cite{sudo24} and another in {\async} as in \cite{KshemkalyaniS21-OPODIS}.  
Table \ref{tbl: Comparative_analysis-rooted} compares and contrasts our results with the state-of-the-art.  We say initial configuration {\em rooted}, if all $k$ robots are initially at the same node, {\em general} otherwise.  
\begin{itemize}
    \item Under {\sync} as in Sudo, Shibata, Nakamura, Kim, and Masuzawa \cite{sudo24} [DISC'24], we present the first time optimal $O(k)$-round algorithm with memory complexity $O(\log (k+\Delta))$ bits, which is an $O(\log^2 k)$ factor improvement on \cite{sudo24}.
    \item Under {\async} as in Kshemkalyani and Sharma \cite{KshemkalyaniS21-OPODIS} [OPODIS'21], we present the first algorithm with time complexity $O(k\log k)$ epochs  and memory complexity $O(\log (k+\Delta))$ bits. This time complexity is optimal within an $O(\log k)$ factor and a substantial improvement over  $O(\min\{m,k\Delta\})$ epochs in \cite{KshemkalyaniS21-OPODIS}.
\end{itemize}

The above results are possible from the two techniques, one for {\sync} and another for {\async}, we develop in this paper, which we describe in a nutshell, in Section \ref{section:techniques} (Overview of Challenges and Techniques). Our developed technique for {\sync} finds a fully unsettled empty neighbor (if exists) in $O(1)$ rounds and handles meetings of two DFSs with $O(1)$ overhead to achieve $O(k)$-round algorithm with $O(\log (k+\Delta))$ memory; a {\em fully unsettled} node is the one that has no agent positioned on it yet. Our developed technique for {\async} finds a fully unsettled empty neighbor (if exists) in $O(\log k)$ epochs and handles meetings of two DFSs with $O(1)$ overhead to achieve $O(k\log k)$-epoch algorithm with $O(\log (k+\Delta))$ memory.

\noindent{\bf Paper organization.} We discuss model and some preliminaries in Section \ref{sec:model}. In Section~\ref{sec: relatedwork}, we discuss the related work. We overview challenges and techniques developed to achieve the claimed bounds  
in Section \ref{section:techniques}.  
We build some techniques crucial for our {\sync} and {\async} algorithms in Section \ref{section:emptynodes}. 
We discuss the {\sync} and {\async} rooted dispersion algorithms in Sections \ref{section:RootedSync} and \ref{section:RootedAsync}, respectively. We then discuss algorithms handling general initial configurations in Section \ref{section:general_dispersion}.  Finally, we conclude in Section \ref{sec: conclusion}.

\section{Model, Notations, and Preliminaries}\label{sec:model}
\noindent{\bf Graph.} Let $G = (V, E)$ be a simple, undirected, and connected arbitrary graph with $n = |V|$ nodes, $m = |E|$ edges, and diameter $D$.   We denote the set of \emph{neighbors} of a node $v \in V$ by $N(v) = \{u \in V \mid \{u, v\} \in E\}$. We denote the {\em degree} of  $v\in V$ by $\delta_v = |N(v)|$. We have that $\Delta = \max_{v \in V} \delta_v$, i.e., $\Delta$ represents the {\em maximum degree} of $G$. The graph $G$ is {\em anonymous}, meaning that the nodes in $V$ do not have identifiers. However, $G$ is {\em port-labeled} meaning that the edges incident to a node $v$ are locally labeled at $v$ such that an agent located at $v$ can uniquely distinguish these edges. Specifically, these edges are assigned distinct labels $1, 2, \dots, \delta_v$ at node $v$. We refer to these local labels as \emph{port numbers}. 
The port number at $v$ for an edge $\{v, u\}$ is denoted by $p_v(u)$. Since each edge $\{v, u\}$ has two endpoints, it is assigned two labels, $p_v(u)$ and $p_u(v)$, at nodes $v$ and $u$, respectively. Note that these labels are independent, meaning $p_u(v) \neq p_v(u)$ may hold, and hence there is no correlation between port numbers assigned for any two nodes $u,v\in G$. For any $v \in V$, we define $N(v, i)$ as the node $u \in N(v)$ such that $p_v(u) = i$. For simplicity, we define $N(v, \bot) = v$ for all $v \in V$. Each node $v\in V$ is memory-less and cannot retain any information.

\noindent{\bf Agents.} We consider the set $\A=\{a_1,\ldots,a_k\}$ of $k\leq n$ agents positioned initially on the nodes of $G$. 
The agents have identifiers, i.e., each agent $a_i$ has a positive integer as its identifier $a_i.\id$. The identifiers are unique meaning that $a_i.\id \neq a_j.\id$ for any $a_i, a_j \in \A, j\neq i$. 
We assume that $a_i.\id\in [1,k^{O(1)}]$.

In {\async}, an agent $a_i$ might see another agent $a_j$ on an edge of $G$ but when an epoch (defined formally later) finishes, it is guaranteed that no agent is on an edge, i.e., all agents are on the nodes of $G$. In {\sync}, no agent sees another agent at any edge of $G$.
Each agent $a_i$ positioned at node $v$ has a read-only variable $a_i.\pin \in \{1, 2, \dots, \delta_v\} \cup \{\bot\}$. The execution starts at time $t=0$. At $t=0$, $a_i.\pin = \bot$ holds. For any time $t \geq 1$, if $a_i$ moves from $v$ to $u$, then $a_i.\pin$ is set to $p_u(v)$ (the port of $u$ incoming from $v$) at the beginning of round $t$. 
We refer to the value of $a_i.\pin$ as the \emph{incoming port} of $a_i$.
The values of all variables in the agent $a_i$, including its identifier $a_i.\id$ and the special variables $a_i.\pin$ and $a_i.\pout$, constitute the state of $a_i$. 
The $a_i.\pout$ variable denotes the {\em outgoing port} which, for agent $a_i$ current at node $v$, is the port used by $a_i$ to exit $v$.   
We use notations $\alpha(a_i)$ and $\alpha(w)$ to denote node where agent $a_i$ resides and agent that resides at node $w$, respectively, i.e., for agent $a_i$ currently residing at node $w$, $\alpha(a_i)=w$ and $\alpha(w)=a_i$.

\noindent{\bf Communication model.} We consider the local communication model in which two agents $a_i,a_j$ can exchange information at any time if they both are at the same node, otherwise, they could not do so. Another communication model is global \cite{KshemkalyaniMS22} in which $a_i,a_j$ can exchange information at any time, irrespective of their location (same node or different).

\noindent{\bf Time cycle.} At any time, an agent $a_i$
 could be active or inactive. When 
 becomes active, $a_i$ performs the ``Communicate-Compute-Move'' (CCM) cycle as follows.
 \begin{itemize}
     \item {\bf Communicate:} Agent $a_i$ positioned at node $u$ can observe the memory of another agent $a_j$ positioned at node $u$. Agent $a_i$ can also observe its own memory.
    \item {\bf Compute:}  Agent $a_i$ may perform an arbitrary computation using the information observed during the ``communicate'' portion of that cycle. This includes the determination of a port to use to exit $u$ and the information to store in the agent $a_j$ that is at $u$.
    \item {\bf Move:} At the end of the cycle,  $a_i$ writes new information (if any) in the memory of an agent $a_j$ at $u$, and exits $u$ using the computed port and reaches a neighbor of $u$.
 \end{itemize}

\noindent{\bf Round, epoch, time, and memory complexity.}
In {\sync}, every agent performs each CCM cycle in synchrony becoming active in every CCM cycle. Therefore, time complexity is measured in rounds (a cycle is a round). 
In {\async}, agents have no common notion of time. No limit exists on the number and frequency of CCM cycles in which an agent can be active except that every agent is active infinitely often. 
We use the idea of an epoch to measure time complexity. An {\em epoch} is the smallest time interval within which each robot is active at least once \cite{Cord-LandwehrDFHKKKKMHRSWWW11}. Formally, let $t_0$ denote the start time. Epoch $i\geq 1$ is the time interval from $t_i-1$ to $t_i$ where $t_i$ is the first time instant after $t_i-1$ when each agent has finished at least one complete LCM cycle. Therefore, for {\sync}, a round is an epoch. 
We will use the term ``time'' generically to mean rounds for {\sync} and epochs for {\async}.
Memory complexity is the number of bits stored at any agent over one CCM cycle to the next. The temporary memory needed during the Compute phase is considered free.

\noindent{\bf Initial configurations.} Since agents are assumed to be positioned arbitrarily initially, there may be the case that all $k\leq n$ agents are at the same node, which we denote as {\em rooted initial configuration}. Any initial configuration that is not rooted is denoted as {\em general}. In any general initial configuration, agents are on at least two nodes. A special case of general configuration is a {\em dispersion configuration} in which $k$ agents are on $k$ different nodes.

\section{Related Work}\label{sec: relatedwork}
The state-of-the-art results on dispersion are listed in Table \ref{tbl: Comparative_analysis-rooted} as well as the established results. 
The state-of-the-art in {\sync} is the result due to \cite{sudo24} appeared in [DISC'24], whereas the result due to \cite{KshemkalyaniS21-OPODIS} appeared in [OPODIS'21] is the state-of-the-art in {\async}. 
These two results solve the problem starting from any initial configuration (rooted or general). 

For the rooted initial configurations, some better bounds were obtained.  In {\sync}, there are two results \cite{sudo24} [DISC'24]: (i) time complexity optimal $O(k)$ with memory complexity $O(\Delta+\log k)$ and (ii) time complexity $O(k\log k)$ with memory complexity $O(\log (k+\Delta))$. Our {\sync} result achieves time optimality keeping memory $O(\log (k+\Delta))$. In {\async}, there is one result \cite{KshemkalyaniMS22} [JPDC'22]: time complexity $O(D\Delta(D+\Delta))$ with memory complexity $O(D+\Delta\log k)$. Our {\async} result makes time optimal within an $O(\log k)$ factor keeping memory $O(\log (k+\Delta))$. 

The state-of-the-art results discussed in the previous two paragraphs were the culmination of the long series of work \cite{Augustine:2018,BKM24, BKM25,ChandKMS23,DasCALDAM21,GorainSSS22,KaurD2D23,ItalianoPS22,KshemkalyaniICDCN19,KshemkalyaniALGOSENSORS19,KshemkalyaniWALCOM20,KshemkalyaniICDCS20,KshemkalyaniS21-OPODIS,tamc19,SaxenaK025,Saxena025,ShintakuSKM20}. 
The majority of works considered the faulty-free case, except \cite{BKM25,MollaIPDPS21,MollaMM21} which  considered {\em Byzantine} faults (where agents might act arbitrarily)  and \cite{BKM24,BKM25,ChandKMS23,Pattanayak-WDALFR20} considered {\em crash} faults (where some agents might stop working permanently at any time). Moreover, most of the works considered the {\em local communication} model where only agents co-located at a node can communicate at any epoch/round, except Kshemkalyani {\it et al.} \cite{KshemkalyaniMS22} who considered the {\em global communication} model in which agents can communicate irrespective of their positions on the graph at any epoch/round. 
Furthermore, most of the works considered static graphs, except  \cite{KshemkalyaniICDCS20,SaxenaK025,Saxena025} which considered different models of dynamic graphs. Moreover, most of the works presented deterministic algorithms except \cite{DasCALDAM21,tamc19} where randomness is used to minimize the memory complexity. Dispersion is considered with restricted local communication among co-located robots in \cite{GorainSSS22}. Reaching restricted final configurations, such as no two adjacent nodes can contain agents, are considered in \cite{KaurD2D23}. 
The majority of papers considered arbitrary graphs except for some papers where special graphs were considered, for example, grid \cite{BKM24, BKM25, KshemkalyaniWALCOM20}, ring \cite{Augustine:2018,MollaMM21}, and trees \cite{Augustine:2018,KshemkalyaniMS22}.  Finally, results in \cite{ChandKMS23,KshemkalyaniALGOSENSORS19} were established assuming parameters $n,m,\Delta,k$ known to the agents a priori.   

Additionally, dispersion is closely related to graph exploration which has been quite extensively studied 
for specific as well as arbitrary graphs, 
e.g., \cite{Bampas:2009,Cohen:2008,Dereniowski:2015,Fraigniaud:2005,MencPU17}. 
Another related problem is the scattering which has been studied for rings \cite{ElorB11,Shibata:2016} and grids \cite{Barriere2009,Das16,Poudel18}. 
Dispersion is also related to the load balancing problem, where a given
load has to be (re-)distributed among several nodes,
e.g., \cite{Cybenko:1989,Subramanian:1994}.
%
Very recently, a dispersion solution has been used in \cite{KshemkalyaniKMS24} electing a leader among agents as well as solving graph-level tasks such as computing minimum spanning tree (MST), maximal independent set (MIS), and minimal dominating set (MDS) problems.

\section{Overview of Challenges and Techniques}
\label{section:techniques}
\subsection{Challenges} The literature on dispersion relied mostly on search techniques,  depth-first-search (DFS) and breadth-first-search (BFS). DFS has been preferred over BFS since it facilitated optimizing memory complexity along with time complexity. 
DFS time complexity depends on how quickly a fully unsettled empty neighbor node can be found from the current node to settle an agent. Recall that a fully unsettled node is the one that has no agent positioned on it yet.  
Prior to the result of Sudo {\it et al.} \cite{sudo24} appeared in [DISC'24],  the fully unsettled neighbor search needed $O(\Delta)$ time. Additionally, $k-1$ different nodes need to be visited to settle $k$ agents (one per node). Therefore, the dispersion was solved in $O(k\Delta)$ time; precisely $O(\min\{m,k\Delta\})$ time since DFS finishes after examining all the edges in the graph. Kshemkalyani and Sharma \cite{KshemkalyaniS21-OPODIS} proved that this bound applies to both {\sync} and {\async} and is the state-of-the-art bound in {\async}. 
Sudo {\it et al.} \cite{sudo24} developed a clever technique in {\sync} to find a fully unsettled empty neighbor node in $O(\log k)$ time. Using this technique, they achieved dispersion for rooted initial configurations in $O(k\log k)$ time, improving significantly on $O(\min\{m,k\Delta\})$ time complexity.

In general initial configurations, let $\ell$ be the number of nodes with agents on them 
(for the rooted case, $\ell=1$). There will be $\ell$ DFSs initiated from those $\ell$ nodes. It may be the case that two or more DFSs {\em meet}. The meeting situation needs to be handled in a way that ensures it does not increase the time required to find fully unsettled neighbor nodes. Kshemkalyani and Sharma \cite{KshemkalyaniS21-OPODIS} showed that such meetings can be handled in additional time proportional to $O(\min\{m,k\Delta\})$ in both {\sync} and {\async}, and hence overall time complexity remained $O(\min\{m,k\Delta\})$ for general initial configuration, the same time complexity as in rooted initial configurations. Sudo {\it et al.} \cite{sudo24} handled such meetings in {\sync} with $O(\log k)$ factor overhead compared to $O(k\log k)$ time for the rooted initial configurations and hence the dispersion time complexity became $O(k\log^2k)$ for general initial configurations in {\sync}.

DFS starts in the forward phase and works alternating between forward and backtrack phases until $k-\ell$  fully unsettled nodes are visited ($\ell$ nodes already have agents and one on each such node can settle there). During each forward phase from one node to another, one such fully unsettled node becomes settled. To visit $k$ different fully unsettled nodes, even when starting from a rooted initial configuration, a DFS must perform at least $k-1$ forward phases, and at most $k-1$ backtrack phases. Therefore, the best DFS time complexity is $2(k-1)=O(k)$, which is asymptotically optimal since in graphs (such as line) exactly $k-1$ forward phases are needed in the worst-case (consider the case of all $k$ agents are on either end node of the line graph). 
Therefore, the challenge is how to guarantee only $k-1$ forward phases and each forward phase to finish in $O(1)$ time to obtain $O(k)$ time bound. 

In this paper, in {\sync}, we devise techniques to find, from any node, a fully unsettled empty neighbor (if exists) in $O(1)$ rounds and handle meetings of two DFSs with $O(1)$ overhead, achieving time-optimal $O(k)$-round algorithm with $O(\log (k+\Delta))$ memory per agent. 
Despite the efficacy of the developed techniques in {\sync} achieving time optimality with $O(\log (k+\Delta))$ memory, we are not able to extend the {\sync} techniques to {\async}. Nevertheless, we are able to devise techniques to find, from any node, a fully unsettled empty neighbor (if exists) in $O(\log k)$ epochs. Interestingly, our {\async} technique extends the {\sync} technique of Sudo {\it et al.} \cite{sudo24} that achieved $O(k\log^2k)$-round solution for dispersion in {\sync}. Additionally,  our technique handles meetings of two DFSs with $O(1)$ overhead even in {\async}, achieving  $O(k\log k)$-epoch algorithm with $O(\log (k+\Delta))$ memory in {\async}. This result is interesting and important since it answers an open question from Sudo {\it et al.} \cite{sudo24} in the affirmative about whether $o(\min\{m,k\Delta\})$-time algorithm could be designed for dispersion in {\async} with $O(\log (k+\Delta))$ memory.  Furthermore, our {\async} result improves the {\sync} time bound of Sudo {\it et al.} \cite{sudo24} by an $O(\log k)$ factor. It remains open whether our {\sync} technique can be extended to achieve $O(k)$-epoch algorithm with $O(\log (k+\Delta))$ memory in {\async}.   

\subsection{Techniques for \texorpdfstring{\sync}{sync}}
\label{subsection:sync-technique}

As discussed above, to be able to solve dispersion in $O(k)$ rounds, the number of forward phases should be $O(k)$ and each forward phase should take $O(1)$ rounds. To be able to run the forward phase in $O(1)$ rounds at a node $v$,  a fully unsettled neighbor node of $v$ should be found in $O(1)$ rounds, where an agent can settle. 
A simple idea to guarantee $O(1)$ rounds for the forward phase at $v$ is to probe all $\delta_v$ neighbors of $v$ in parallel (which we call {\em synchronous probing}). To probe $\delta_v$ neighbors in parallel, $v$ needs at least $\delta_v$ agents positioned. 
However, let $k'\leq k$ be the number of agents at $v$. 
When $k'<o(\delta_v)$, synchronous probing cannot finish in $O(1)$ rounds. We know that using the probing technique of Sudo {\it et al.} \cite{sudo24},  synchronous probing can be done in $O(\log \delta_v)$ rounds, even when $k'<o(\delta_v)$.  Therefore, dispersion needs $O(k+\Delta \log \Delta)$ time in the worst-case, which is $O(k \log k)$,  when $\Delta=\Theta(k)$, matching Sudo {\it et al.} \cite{sudo24}.

Our proposed technique runs synchronous probing in $O(1)$ rounds. Particularly, we guarantee that, during DFS, $\lceil k/3\rceil$ agents are available for synchronous probing. 
We make $\lceil k/3\rceil$ agents available leaving $\geq \lceil k/3\rceil$ fully unsettled nodes empty during DFS, i.e., the agents cannot settle on $\geq \lceil k/3\rceil$ nodes of the associated DFS tree $T_{DFS}$ until DFS finishes.  Therefore, the $\lceil k/3\rceil$ agents can finish synchronous probing at any node  in $O(1)$ rounds. 
If node $v$ has $\delta_v\leq \lceil k/3\rceil$, then $\lceil k/3\rceil$ agents can probe all $\delta_v$ neighbors of $v$ in 2 rounds. For $\delta_v>\lceil k/3\rceil$, notice that  probing only $\min\{k,\delta_v\}$ neighbors of $v$ is enough to find a fully unsettled empty neighbor of $v$ (if exists). This is because no more than $k$ nodes are going to be occupied with agents at any time since there are $k$ agents. Therefore, with $\lceil k/3\rceil$ agents, probing at $v$ can finish in at most 3 iterations, i.e.,  $O(1)$ rounds.

We provide an illustration of synchronous probing in Fig.~\ref{fig:sync-probe}.  Although having $\lceil k/3\rceil$ agents helps in finishing probing in $O(1)$ rounds, guaranteeing the availability of $\lceil k/3\rceil$ agents creates three major challenges Q1--Q3 below. 
We need some notations. Suppose each agent maintains a variable $\settled \in \{\bot,\top\}$. Initially, the $\lceil k/3\rceil$ agents start as   {\em seeker} and the remaining $\lfloor 2k/3\rfloor$ agents start as {\em explorer}.
 We say an explorer agent $a_i$ a \emph{settler} when $a_i.\settled = \top$. 
 The seeker agents are used in synchronous probing and they settle only after synchronous probing is not needed anymore.

\begin{itemize}
    \item [{\bf Q1.}] {\bf Which DFS tree nodes to leave empty?} 
    Since we allocate $\lceil k/3\rceil$ agents as seeker agents, they cannot settle during DFS. Additionally, the DFS needs to visit $k$ nodes, i.e., the DFS tree $T_{DFS}$ should have size $k$. For this, we need to leave at least $\lceil k/3\rceil$ nodes of $T_{DFS}$ empty. The challenge is how to meet such requirements.  
    We carefully leave the nodes of $T_{DFS}$ empty such that each such empty node has a settler within 2 hops in $T_{DFS}$. 
    The idea is to leave the nodes of $T_{DFS}$ at odd depth (depth of the root is $0$) empty. This is sufficient when no node in $T_{DFS}$ is a branching node (i.e., a degree more than 2). For branching nodes (depending on their depth), we carefully decide on whether to (i) put extra settlers on their empty children nodes or (ii) remove some settlers from their non-empty children nodes.    
   
    We prove that our approach leaves at least $\geq \lceil k/3\rceil$ nodes of $T_{DFS}$ empty until DFS finishes. Therefore, DFS finishes with having $k$ nodes in $T_{DFS}$ but with only at most $\lfloor 2k/3\rfloor$ nodes of $T_{DFS}$ occupied with settlers. This guarantees $\lceil k/3\rceil$ seeker agents we allocated are available to do synchronous probing until DFS finishes. We discuss the algorithm in detail in Section \ref{section:emptynodes} (Algorithm to Select Empty Nodes) along with illustrations (Fig.~\ref{fig:empty-tree-arbitrary-T}).

    \item [{\bf Q2.}] {\bf How to successfully run DFS despite being some of its tree nodes empty?} Suppose a settler $a_j$ is positioned on a node of $T_{DFS}$. Let that node be $a_j$'s {\em home node}. We ask settler $a_j$ to {\em oscillate} in a round-robin manner covering the empty nodes starting from and ending at its home node. We call settler $a_j$ that oscillates an {\em oscillating settler}. We guarantee that an oscillating settler at depth $d$ (the root of $T_{DFS}$ is at depth $0$) only needs to cover at most 3 empty children nodes at depth $d+1$  or at most 2 sibling nodes at the same depth $d$ in $T_{DFS}$. This guarantees the matching of settlers to empty nodes such that each settler's round-robin oscillation trip finishes in (at most) $6$ rounds. Some of the settlers may never oscillate and they are called {\em non-oscillating settlers}. A non-oscillating settler never leaves its home node. 
    Some settlers may initially be non-oscillating, become oscillating over time, and finally become non-oscillating. 
    When an oscillating settler is not at its home node, we say that it is at its {\em oscillating home node}. We provide illustrations in Figs.~\ref{fig:oscillation1} and \ref{fig:oscillation2}.   
    \item [{\bf Q3.}] {\bf How to settle seeker agents to the empty nodes after DFS finishes?} After having $k$ nodes in $T_{DFS}$, the DFS finishes. The $\lceil k/3\rceil$ seeker agents (and remaining explorers, if any) first go to the root of $T_{DFS}$ as a group following the parent pointers in $T_{DFS}$. 
    Then, they re-traverse $T_{DFS}$, starting from the root node, to position themselves on the empty nodes in $T_{DFS}$.
      We make this re-traversal process finish in $O(k)$ time with $O(\log (k+\Delta))$ memory through the use of a {\em sibling pointer} technique which we develop. 
      Without sibling pointers, the memory requirement for re-traversal becomes $O(\Delta\log(k+\Delta))$\footnote{This memory bound can be made $O(\Delta+\log k)$ by just storing the bit information to separate the ports at a node belonging to the DFS tree $T_{DFS}$, instead of storing the port numbers themselves. This is exactly what Sudo {\it et al.} \cite{sudo24} did to have memory $O(\Delta+\log k)$ for their {\sync} rooted dispersion algorithm running in $O(k)$ rounds. Through our sibling pointer technique, we were able to keep the time complexity $O(k)$ with memory only $O(\log(k+\Delta))$.}.
    
    An oscillating setter, once its job is done, stops oscillating, returns to its home node, transforms itself as a non-oscillating settler and settles.     
\end{itemize}

\noindent{\bf Handling general initial configurations.}
So far we discussed techniques to achieve $O(k)$ time complexity for rooted initial configurations. In general initial configurations, 
there will be $\ell$ DFSs initiated from $\ell$ nodes ($\ell$ not known). Each DFS follows the approach as in the rooted case. Let a node has $k_1$ agents running DFS $D_1$. We show that having $\geq \lceil k_1/3\rceil$ agents as seekers is enough to finish $D_1$ if $D_1$ does not meet any other DFS, say $D_2$. This provides the guarantee that, if no two DFSs ever meet, dispersion is achieved. In case of a meeting, we develop an approach that handles the meeting of two DFSs $D_1$ and $D_2$ with overhead the size of the larger DFS between $D_1$ and $D_2$. In other words, $k_1+k_2$ agents belonging to $D_1$ and $D_2$ disperse in $O(k_1+k_2)$ rounds, if a meeting with the third DFS $D_3$ never occurs during DFS. If that meeting occurs, we show that the time complexity becomes $O(k_1+k_2+k_3)$ rounds. Therefore, the worst-case time complexity starting from any $\ell$  nodes becomes $O(k)$ rounds. 
Specifically, to achieve this runtime, we extend the {\em size-based subsumption} technique of Kshemkalyani and Sharma \cite{KshemkalyaniS21-OPODIS} [OPODIS'21] which works as follows. 
Suppose DFS $D_1$ meets DFS $D_2$ at node $w$ (notice that $w$ belongs to $D_2$). Let $|D_i|$ denote the number of agents settled from DFS $D_i$ (in other words, the number of nodes in the DFS tree $T_{D_i}$ built by $D_i$ so far). $D_1$ subsumes $D_2$ if and only if $|D_2|<|D_1|$, otherwise $D_2$ subsumes $D_1$. The agents settled from subsumed DFS are collected and given to the subsuming DFS to continue with its DFS, which essentially means that the subsumed DFS does not exist anymore.  This subsumption technique guarantees that one DFS out of $\ell'$ met DFSs (from $\ell$ nodes) always remains subsuming and grows monotonically until all agents settle.

\subsection{Techniques for \texorpdfstring{\async}{}}

Although the techniques discussed in Section \ref{subsection:sync-technique} provided optimal $O(k)$-round solution in {\sync} with $O(\log (k+\Delta))$ memory, we are not able to make it work to obtain the same $O(k)$-epoch solution in {\async}. However, 
we are able to extend the {\sync} technique of Sudo {\it et al.} \cite{sudo24} to {\async} such that finding a fully unsettled neighbor node and moving to that neighbor finishes in $O(\log k)$ epochs. This extension indeed shows that the probing technique of Sudo {\it et al.} \cite{sudo24} is not inherently dependent on synchrony assumption.  

We first discuss the challenge for such an extension and how we overcome it. Suppose the DFS is currently at node $w$ with at least three agents $a_1,a_2,a_3$ on it, $a_1$ settled at $w$ and $a_2,a_3$ need to be settled on other nodes  To settle $a_2$, a fully unsettled neighbor node of $w$ needs to be found. The idea of Sudo {\it et al.} \cite{sudo24} for {\sync} is as follows. $a_2$ leaves $w$ via port-1 to visit port-1 neighbor. 
Suppose $a_2$ returns with the knowledge that port-1 neighbor is empty, which is in fact the fully unsettled neighbor node of $w$ due to {\sync}. The DFS then makes a forward move to port-1 neighbor, $a_2$ settles at that node, and the DFS continues at that node to find a fully unsettled node to settler $a_3$. 
Suppose $a_2$ returns with the knowledge that port-1 neighbor is non-empty (has an agent settled). While returning to $w$, $a_2$ brings that agent (say $a_{p1}$) to $w$. Next time, $a_2$ and $a_{p1}$ visit in parallel port-2 and port-3 neighbors of $w$ and return either with settled agents at those neighbors or empty-handed. This process continues until possibly $\min\{k,\delta_w\}$ neighbors are visited. This finishes in $O(\log \min\{k,\delta_w\})=O(\log k)$ rounds since every next iteration of probing is done by double the number of agents. 

The question is how to run this approach in {\async}. We synchronize agents running probing by waiting for all the agents doing probing in the current iteration to return to $w$ to start the next iteration. This can be easily done since $w$ can have a count on how many are left and how many are returned. Fig.~\ref{fig:async-probe} illustrates these ideas. 

Now suppose a fully unsettled neighbor is found (if exists) or all neighbors are visited and no empty neighbor is found. The question is how to return the settled helper agents back to their home nodes. The {\em settled helper agents} are the settled agents at $w$'s neighbors brought at $w$ to help with probing. 
In {\sync}, after a fully unsettled neighbor node is found, $a_2,a_3$ go to that neighbor, say $v$, where $a_2$ settles and the settled helper agents go to their respective nodes, which finished in a round. However, in {\async}, it may not be the case, i.e., it might hamper the search of $a_3$ to find a fully unsettled neighbor node of $v$ (if exists) to settle. To illustrate, let $u$ be the common neighbor of $w$ and $v$. Let the agent $a_u$ originally settled at $u$ is in transit from $w$ to $u$ and has not arrived at $u$ yet. Let $a_3$ start probing the neighbors from $v$ and visit $u$ before $a_u$ reaches $u$. $a_3$ finds $u$ as a fully unsettled neighbor node (since empty) and decides to settle at $u$, violating the dispersion configuration, since the settled helper agent $a_u$ which is in transit to its home node $u$ arrives at $u$, there will be two settlers at $u$. 

We overcome this difficulty as follows. Before, $a_2$ leaves $w$ to settle at $v$, we guarantee that settled helper agents reach their respective home nodes, which are the neighboring nodes of $w$. Let there be $\alpha\leq \min\{k,\delta_w\}$ settled helper agents at $w$ after probing finishes. We pair agents and send a pair each at $\lfloor \alpha/2\rfloor$ neighbors of $w$. For each neighbor, one agent for which it is its home stays at that neighbor and one returns to $w$, which guarantees that  $\lfloor \alpha/2\rfloor$ settled helper agents now reached their nodes, despite asynchrony. We wait until one robot each from  $\lfloor \alpha/2\rfloor$ pairs returns to $w$. We then send 2 agents each to $\lfloor \alpha/4\rfloor$ neighbors of $w$ and $\lfloor \alpha/4\rfloor$ agents return. Finally, repeating this procedure,  there will be only one settled helper agent at $w$ which will be settled at its node by sending it with $a_2$ and $a_2$ returns to $w$.  Fig.~\ref{fig:guestseeoff} illustrates these ideas. This process finishes in $O(\log k)$ epochs. Therefore, when $a_3$  runs its procedure to find a fully unsettled neighboring node and that procedure finds an empty neighbor of $v$, then that empty node is indeed fully unsettled. Thus, we can solve dispersion in {\async} for rooted initial configurations in $O(k\log k)$ epochs. 

\vspace{2mm}
\noindent{\bf Handling general initial configurations.} We proceed with the technique used in our {\sync} general initial configurations to handle general initial configurations in {\async}. Interestingly, the size-based subsumption technique of Kshemkalyani and Sharma \cite{KshemkalyaniS21-OPODIS}
works even in {\async} with $O(\log(k+\Delta))$ memory. Since we have an algorithm for rooted initial configuration running in $O(k\log k)$ epochs, we merge these two ideas to achieve $O(k\log k)$-epoch solution for the general initial configurations in {\async} with $O(\log(k+\Delta))$ memory.

\section{Empty Nodes, Oscillation, and \texorpdfstring{\sync}{} and \texorpdfstring{\async}{} Probing}
\label{section:emptynodes}
We first discuss empty node selection and oscillation techniques used in our {\sync} algorithms. We then discuss neighbor probing techniques for both {\sync} and {\async} algorithms.

\subsection{Empty Node Selection}
We describe here an algorithm \texttt{Empty\_Node\_Selection()} (pseudocode in Algorithm \ref{algorithm:emptynode}) that decides which nodes of $T_{DFS}$ to leave empty in our {\sync} algorithms (Sections \ref{section:RootedSync} and \ref{sec:gensync}). Let $T$ be an arbitrary $T_{DFS}$ with root $r$. We assume here that $T$ is known and \texttt{Empty\_Node\_Selection()} is a centralized algorithm.  We later show how \texttt{Empty\_Node\_Selection()} can be implemented during DFS when $T_{DFS}$ is monotonically growing in its size. 
Let the depth of $T$ be $d_{\max}$ (root $r$ at depth $0$).

\begin{algorithm}[t]
\caption{\texttt{Empty\_Node\_Selection()}}\label{algorithm:emptynode}
\SetKwInOut{Input}{input:}
\SetKwInOut{Output}{Output:}
{\bf input:} an arbitrary tree $T$ of size $k$ and $k$ agents\\
{\bf output:} settle $\leq \lfloor \frac{2k}{3}\rfloor$ agents at $\leq \lfloor \frac{2k}{3}\rfloor$ nodes of $T$ and leave $\geq \lceil \frac{k}{3}\rceil$ nodes of $T$ empty\\
$d_{\max} \leftarrow$ depth of $T$\\
$lf_i\leftarrow$ a leaf node in $T$\\
$d_{lf_i}\leftarrow$ depth of left node $lf_i$\\ 
settle an agent each on all the nodes of $T$ at even depth\\
\For{each $lf_i$ with a settler}
{
$lf_{i-1}\leftarrow$ parent of $lf_i$\\
$x\leftarrow$ the number of children of $lf_{i-1}$ which are leaves of $T$ like $lf_i$ (including $lf_i$)\\
remove settlers from $\lfloor \frac{2x}{3}\rfloor$ of the leaf children of $lf_{i-1}$\\ 
}
$lf_{even}\leftarrow$ non-leaf node in $T$ at any even depth\\
\For{each $lf_{even}$ with a settler}
{
$x\leftarrow$ the number of children of $lf_{even}$\\
\If{$x>3$}
{
put one settler each  on $\lceil \frac{x-3}{3}\rceil$ children of $lf_{even}$\\ 
}
}
\end{algorithm}

\begin{figure}[H]
    \centering
    \includegraphics[width=0.49\linewidth]{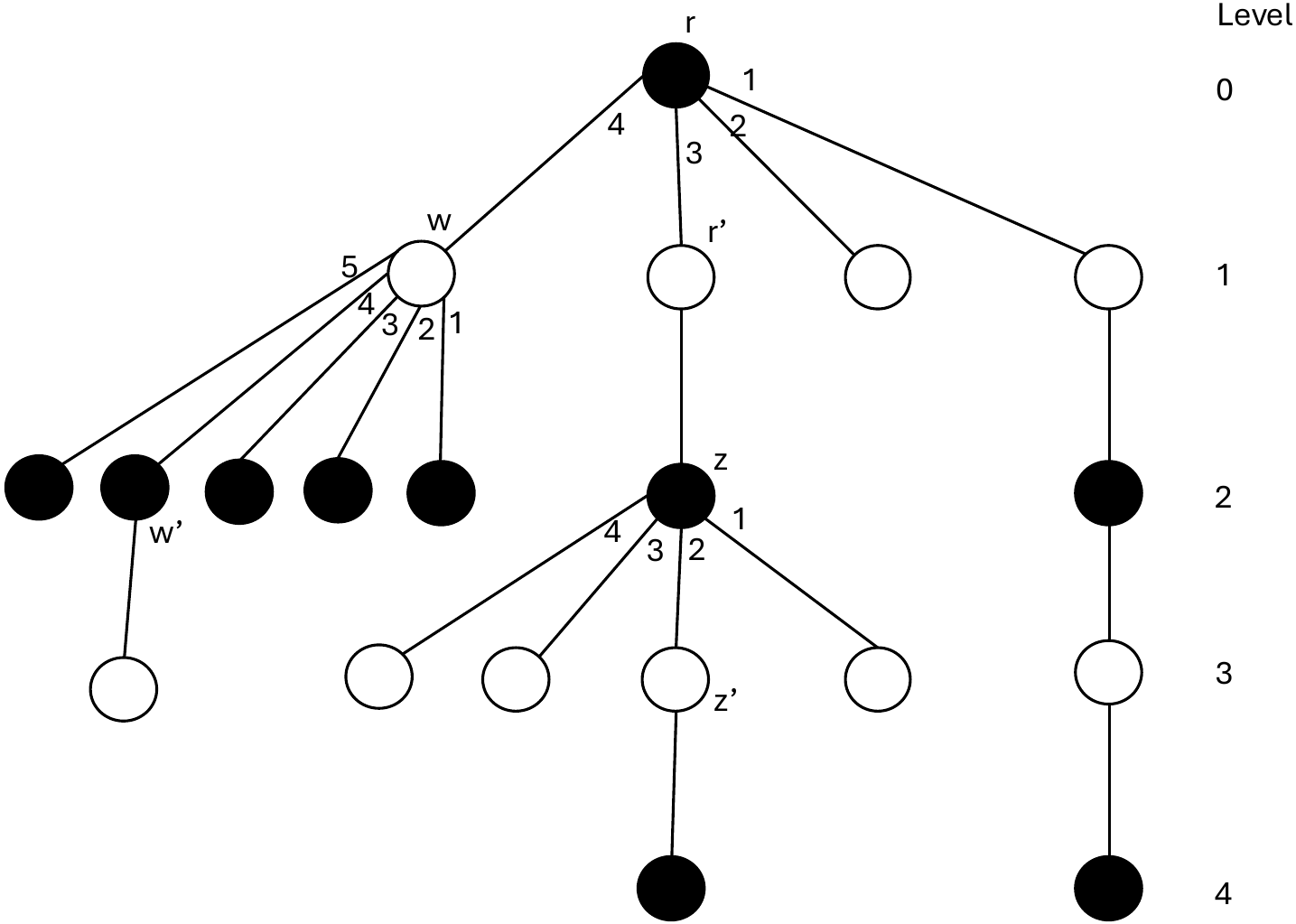}
    \includegraphics[width=0.49\linewidth]{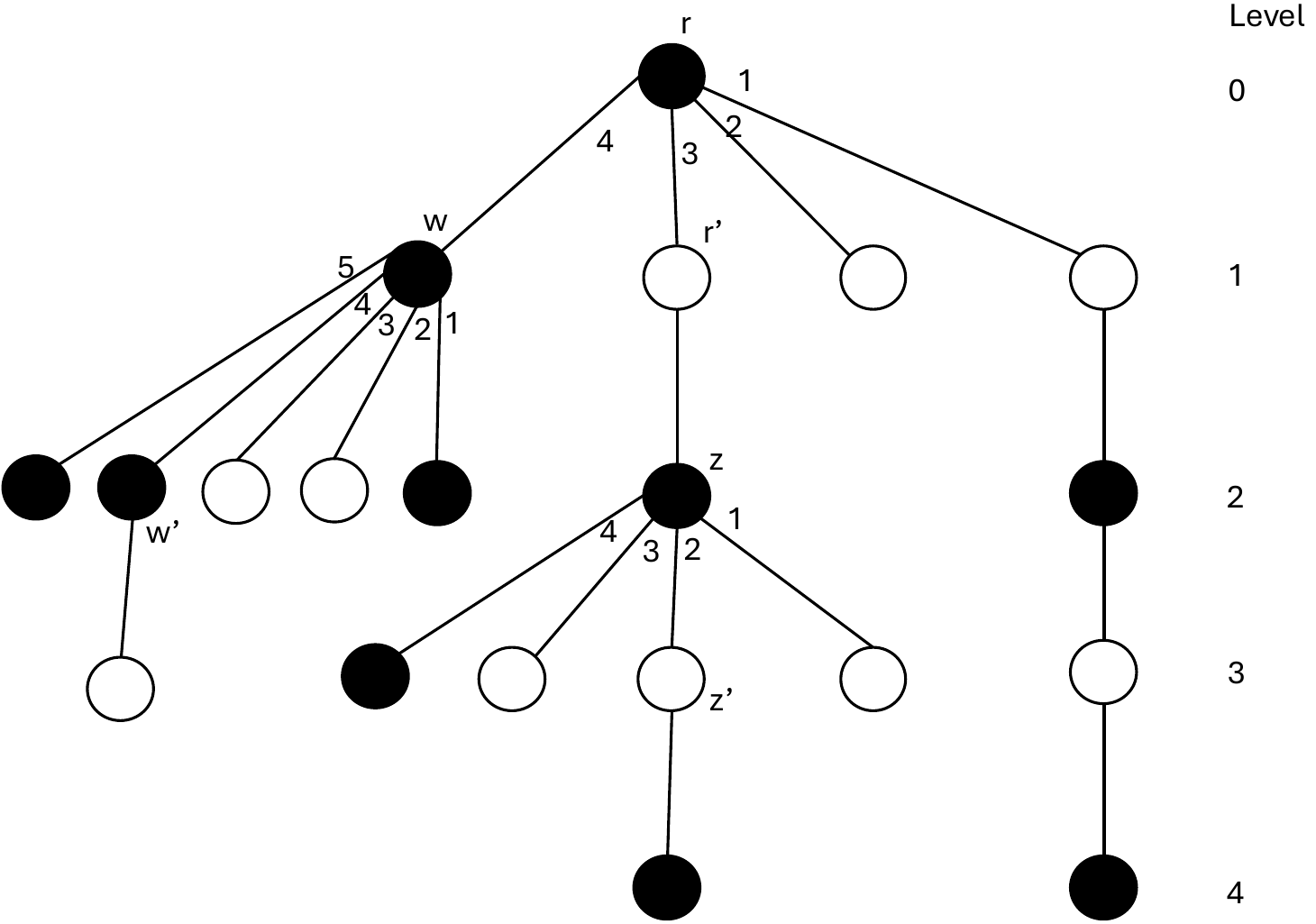}
    \caption{An illustration of which nodes of an arbitrary tree $T$ are occupied with a settler and which nodes are left empty. ({\bf left}) Tree $T$ after nodes at even depth are occupied with a settler. ({\bf right}) Tree $T$ after adjustment through removing some settlers and putting new settlers.}
    \label{fig:empty-tree-arbitrary-T}
\end{figure}

The algorithm is as follows.
We put a settler agent each on all the nodes of $T$ at depth $0,2,4,\ldots$. We then look at the nodes of $T$ and either (A) remove extra settlers or (B) put new settlers. Fig.~\ref{fig:empty-tree-arbitrary-T} provides an illustration. Let $lf_i$ be a leaf node of $T$. Let $lf_i$ has depth $d^i$; the depth of $T$ is $d_{\max}:=\max_{i}d^i.$ 
Depending on whether $d^i$ is even or odd, $lf_i$ has either a settler or is empty.
\begin{itemize}
    \item {\bf Case A -- Remove extra settlers:} Consider the leaves of $T$ which have settlers (the children of node $w$, except $w'$, in Fig.~\ref{fig:empty-tree-arbitrary-T} (left)). Let $lf_i$ be one such leaf on $T$. 
Let $lf_{i-1}$ be a node in $T$  which is the parent of $lf_i$ (node $w$ in Fig.~\ref{fig:empty-tree-arbitrary-T} (left)). Note that  
$lf_{i-1}$ does not have a settler since it is the parent of $lf_i$ with a settler.
Let $x$ be the total number of children of $lf_{i-1}$ that are leaves of $T$ ($x=4$ for $w$ in Fig.~\ref{fig:empty-tree-arbitrary-T} (left)). Note that all these $x$ children must have a settler each. If $x=1$, we do nothing  
(the only child of $z'$ in Fig.~\ref{fig:empty-tree-arbitrary-T} (left)).  
If $x>1$, we remove settlers from $\lfloor \frac{2x}{3}\rfloor$ of those children and hence only  
$\lceil \frac{x}{3}\rceil$ children are left with a settler each. 
For $w$ in Fig.~\ref{fig:empty-tree-arbitrary-T} (right)), we removed two setters from children of $w$ leading through ports 2 and 3 of $w$. 

\item {\bf Case B -- Put new settlers:} Let $lf_{\leq 1}$ be a non-leaf node in $T$ at any even depth from $r$. Notice that  $lf_{\leq 1}$ has a settler since it is at even depth. 
The nodes $r$ and $z$ in Fig.~\ref{fig:empty-tree-arbitrary-T}.
Let $x$ be the number of children of $lf_{\leq 1}$ (pick node $z$). Note that all these $x$ children do not have a settler since they are at odd depth. If $x>3$, we put one settler each on $\lceil \frac{x-3}{3}\rceil$ children of $lf_{\leq 1}$. As shown in Fig.~\ref{fig:empty-tree-arbitrary-T} (right), we put a settler on the child of $z$ leading through port 4 of $z$. 
\end{itemize}
We prove the following for \texttt{Empty\_Node\_Selection()} (Algorithm \ref{algorithm:emptynode}). 

\begin{lemma}
\label{lemma:kover3}
{\it \texttt{Empty\_Node\_Selection()} leaves $\geq \lceil \frac{k}{3}\rceil$ nodes empty in any tree $T$ of size $k\geq 3$. }
\end{lemma}
\begin{proof}
Consider $k=3$. Let $T$ be a line with one endpoint the root $r$. The depth of $T$ is 2. The root node and the depth-2 node will have a settler. Therefore, we have $\geq \lceil k/3\rceil$ nodes empty.  Consider now the case of $k\geq 4$. We show that at least $\lfloor k/2\rfloor$ nodes remain empty. If $T$ is a line, this is immediate. If $T$ is not a line (or when $T$ is a line but the root is not at either end), there must be at least a branch. A branch starts from a node which we call a {\em branching} node.
Order the branching nodes based on their depth from $r$ (including $r$).
Consider the branching node closest to $r$ (including $r$. Let that branching node be $b_i$ at depth $i$. Node $b_i$ has at least two children at depth $i+1$. Since we settle agents at even depth nodes of $T$, if $i$ is even, $b_i$ has a settler and its children are empty. 
Let $x$ be the number of those children. If $x\leq 3$, we do not put any settler. If $x>4$, we divide the children into $\lceil \frac{x-3}{3}\rceil$ groups (one group may have at least one and at most 3 children and assign one settler for each group). Therefore, at least half of the children at depth $i+1$ remain empty for that branching node.  

If $i$ is odd, $b_i$ is empty and its children have a settler each. We classify those children according to whether they are leaves of $T$ or not. For non-leaf children, we do nothing. For leaf children (if there are at least two), we again divide them into groups of three (one group may have at least one and at most 3 children, and others have 3 children each) and remove all the settlers except one from each group.  Therefore, at least half of those nodes at depth $i+1$ remain empty for that branching node. 

Let $n'$ be the empty nodes in $T$ and let $n''=k-n'$ be the non-empty nodes. 
Now we examine the branching nodes and non-branching nodes in combination at level $i$. For the non-branching nodes, there is a child at level $i+1$ (which is empty if $i$ is even and non-empty if $i$ is odd). For the branching nodes at level $i$, we have that at level $i+1$, there are as many empty nodes as non-empty nodes. 
Therefore, for any $T$ from depth 1 up to $d_{\max}$, $n'\geq n''$. Since root $r$ has a settler, we have that $n'\geq n''-1$. Since $n'+n''=k$, the lemma follows for any $k\geq 3$. 
\end{proof}

\subsection{Matching Empty Nodes to Settlers and Oscillation} 
We saw that \texttt{Empty\_Node\_Selection()} leaves $\geq \lceil \frac{k}{3}\rceil$ nodes empty in any arbitrary tree $T$ of size $k$. We now prove a property such that the empty nodes could be covered by  $\leq \lceil \frac{2k}{3}\rceil$ agents settled on the nodes of $T$. Consider an agent $s_d$ positioned at a depth-$d$ node in $T$. 
The agent $s_d$ covers (I) either at most 3 empty children nodes at depth $d+1$ or (II) at most 2 empty sibling nodes at depth $d$.
Fig.~\ref{fig:oscillation1} provides an illustration where it is shown how agents (shown in circle with slanted lines) inside each group of dashed boundaries  cover empty children or sibling nodes inside the group.
For Case $I$, for the agent $s_d$ at depth $d$, let the empty nodes at depth $d+1$ be $a,b,c$. An oscillation trip for agent $s_d$ would be $s_d-a-s_d-b-s_d-c-s_d$. Notice that the order in which empty nodes are visited does not impact the length of the oscillation trip.  For Case $II$, let  $a,b$ be the empty sibling nodes at depth $d$. Let $p(s_d)$ be the parent of $s_d$ at depth $d-1$ (and also of  $a,b$) by construction. An oscillation trip is $s_d-p(s_d)-a-p(s_d)-b-p(s_d)-s_d$. The node $p(s_d)$ (at depth $d-1$) if empty will be covered by a settler agent at depth $d-2$ and hence we do not consider $p(s_d)$ as covered by $s_d$ in its oscillation trip although the oscillation trip of $s_d$ passes through $p(s_d)$. The order on which $a,b$ are visited has no impact on the oscillation trip length for $s_d$.

\begin{lemma}
{\it An oscillation trip for an agent $s_d$ settled at a node at depth $d$ finishes in 6 rounds.}
\end{lemma}
\begin{proof}
   The agent $s_d$ at a node $w$ needs to visit (I) either 3 children nodes of $w$ at depth $d+1$ or (II) 2 sibling nodes at depth $d$. For case I, starting from $w$, a roundtrip of 2 rounds takes $w$ to one empty children and back to $w$. Therefore, at most 3 children are visited in 6 rounds with $s_d$ starting from $w$ and ending at $w$. For case II, let $p(s_d)$ be the parent of $s_d$. $p(s_d)$ must be the parent of the 2 empty sibling nodes $s_d$ needs to visit. $s_d$ reaches $p(s_d)$ in 1 round, it then visits two siblings from $p(s_d)$ in 4 rounds roundtrip. It then returns to $w$ from $p(s_d)$ in 1 round. Therefore, in total 6 rounds for both cases I and II.
\end{proof}

The agent at a node that performs an oscillation trip is an oscillating settler. The agent at a node that does not need to perform an oscillation trip is a non-oscillating settler. In Fig.~\ref{fig:oscillation2}, the black agents are non-oscillating settlers for the given tree $T$. However, while running DFS, $T_{DFS}$ monotonically grows and we have to decide on which nodes to leave empty and which agents to assign for oscillation to cover those empty nodes in real-time. At that time, a settler may start as non-oscillating, it become oscillating, and then eventually after DFS finishes become non-oscillating. Fig.~\ref{fig:oscillation2} provides an illustration in which the agents at nodes $y,w,w''$ which were non-oscillating in Fig.~\ref{fig:oscillation1} now become oscillating since the tree $T$ grew during DFS.  

\begin{lemma}
\label{lemma:oscillating}
{\em A settler $\alpha(w)$ at a node $w$ of $T$ is an oscillating settler if and only if:
\begin{itemize}
    \item $w$ is an even depth non-leaf node.
    \item $w$ is an even depth leaf node  with at least one sibling leaf node not covered by another agent. 
    \item $w$ is an odd depth leaf node with at least one sibling leaf node not covered by another agent. 
\end{itemize}
Otherwise, $\alpha(w)$ is a non-oscillating settler. 
}
\end{lemma}
\begin{proof}
The proof follows from the working principle of \texttt{Empty\_Node\_Selection()}.
If $w$ is an even depth non-leaf node, then it must have a child that is empty. $\alpha(w)$ then covers that child. If $w$ is an even depth leaf node and there is at least one other even depth sibling node, $\alpha(w)$ covers it if it was not covered by another sibling agent.  If $w$ is an odd depth leaf node and there are another sibling leaf node, $\alpha(w)$ needs to cover it if it was not covered by another sibling agent.   The agent that does not satisfy any of the above conditions does not oscillate.      
\end{proof}

We now discuss how we form the groups for oscillation. Fig.~\ref{fig:oscillation-groups} illustrates these ideas. We follow \texttt{Empty\_Node\_Selection()} and arrange them as groups according to the increasing port numbers until the group size limit (3 children or 2 siblings to cover) is reached.

\begin{figure}[!t]
    \centering
    \includegraphics[width=0.7\linewidth]{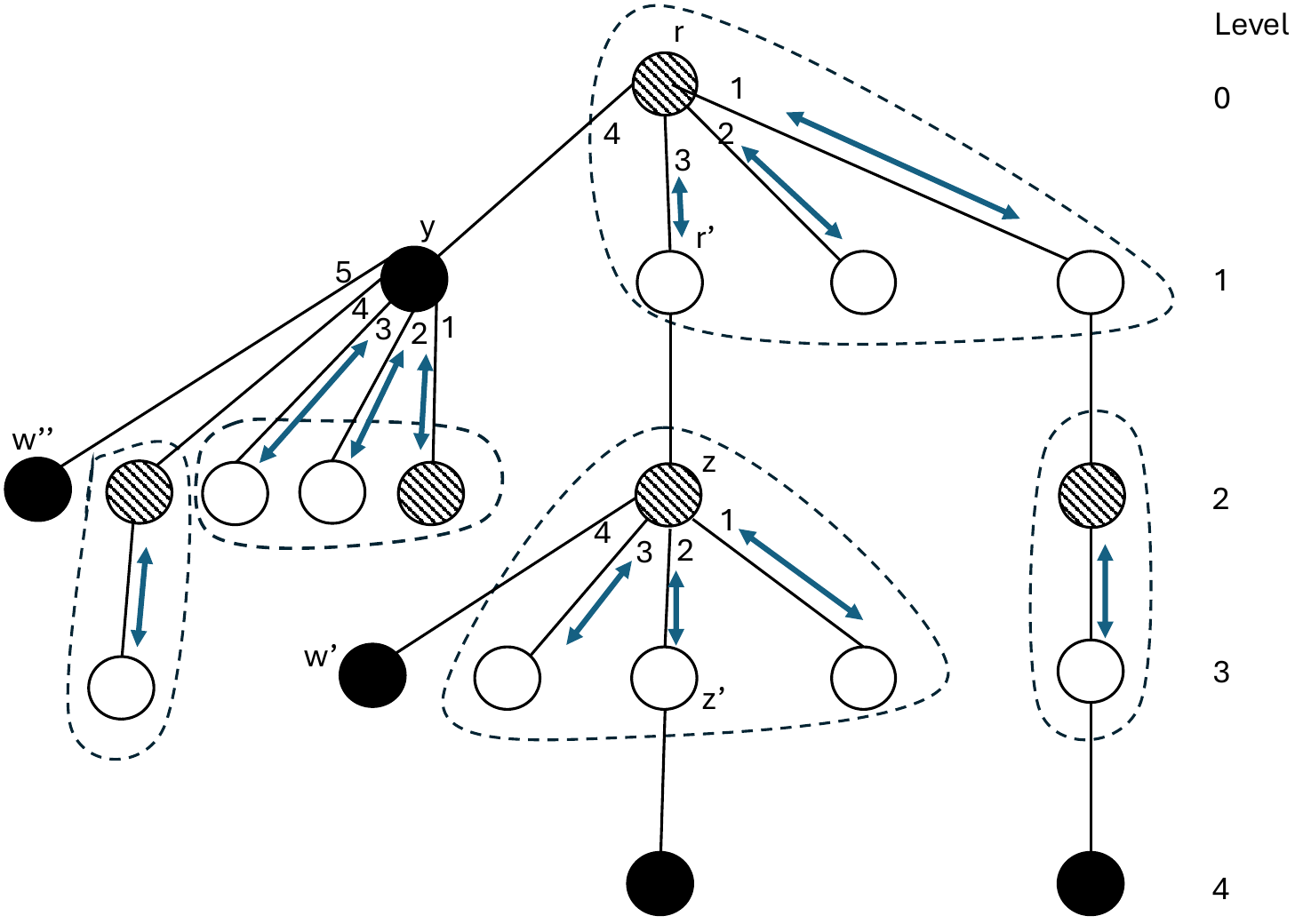}
    \caption{An illustration of how settlers (shown in circle with slanted lines)  cover the empty nodes through oscillation shown as groups inside dashed boundaries. There is exactly one oscillating settler in each group that is responsible for oscillation to cover the empty nodes (at most 2 for the sibling case, otherwise at most 3). Oscillations are two types: (i) An oscillating settler at even depth $d$ covers at most 3 empty nodes at depth $d+1$  (ii) An oscillating settler at even/odd depth $d$ covers at most 2 empty sibling nodes at same depth $d$. The settlers shown as solid black circles are non-oscillating settlers for this figure.}
    \label{fig:oscillation1}
\end{figure}

\begin{figure}[H]
    \centering
    \includegraphics[width=0.4\linewidth]{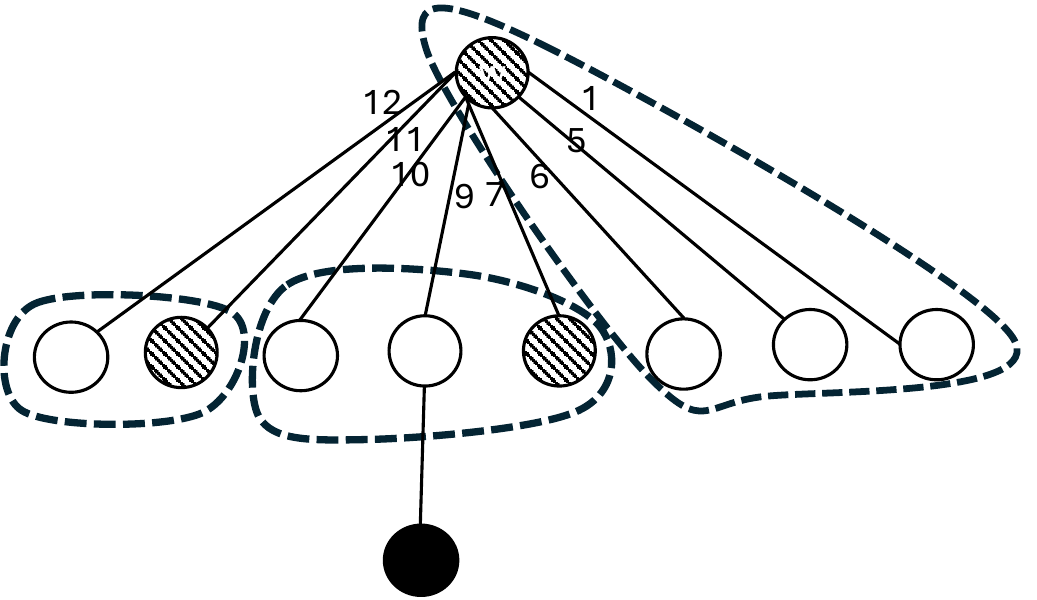}
    \includegraphics[width=0.4\linewidth]{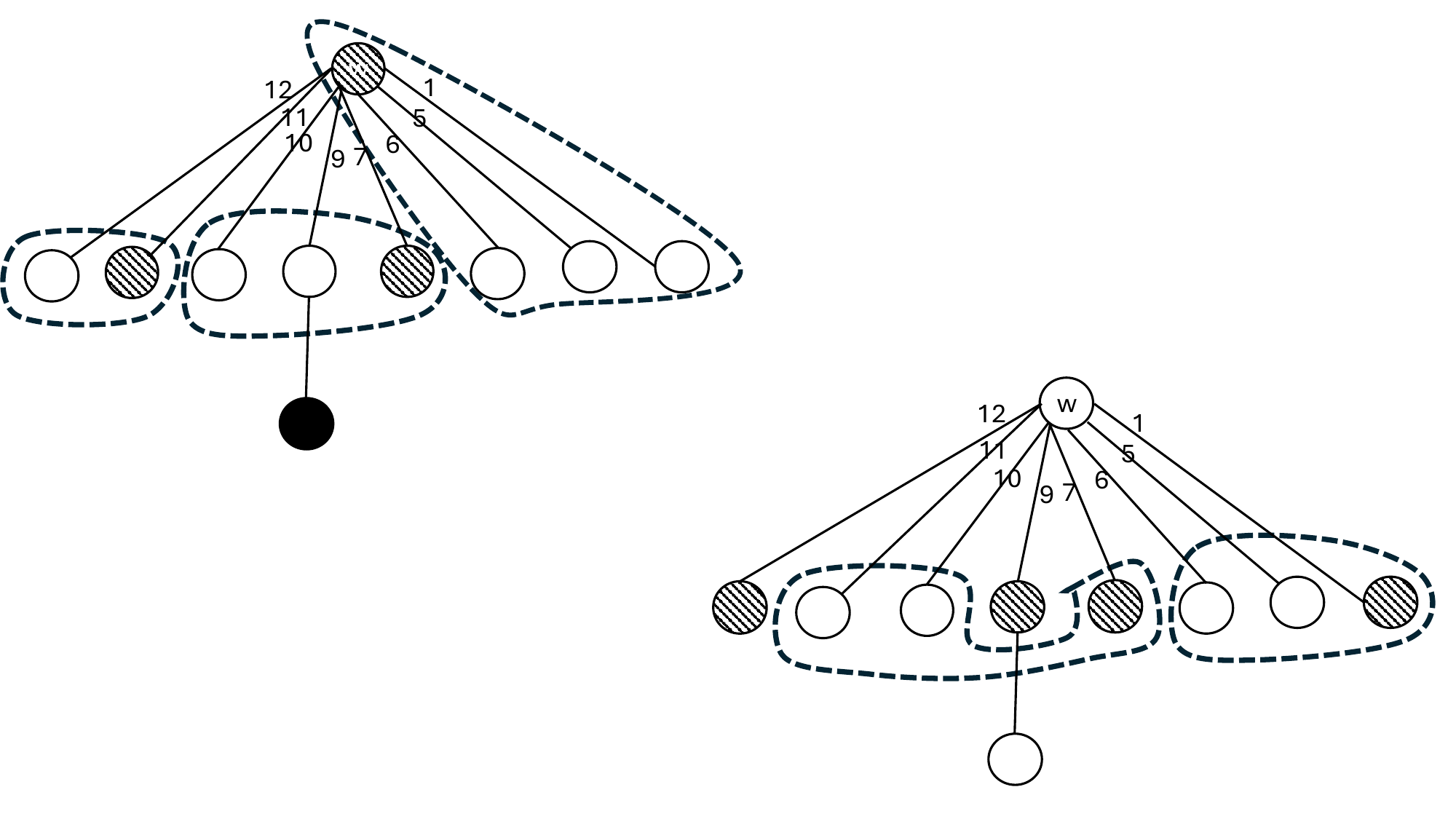}
    \caption{An illustration of how groups are formed for oscillation for a branching node $w$. Groups are formed according to the port numbers in increasing orders for the children that satisfy the oscillation criteria.}
    \label{fig:oscillation-groups}
\end{figure}

\begin{figure}[H]
    \centering
    \includegraphics[width=0.8\linewidth]{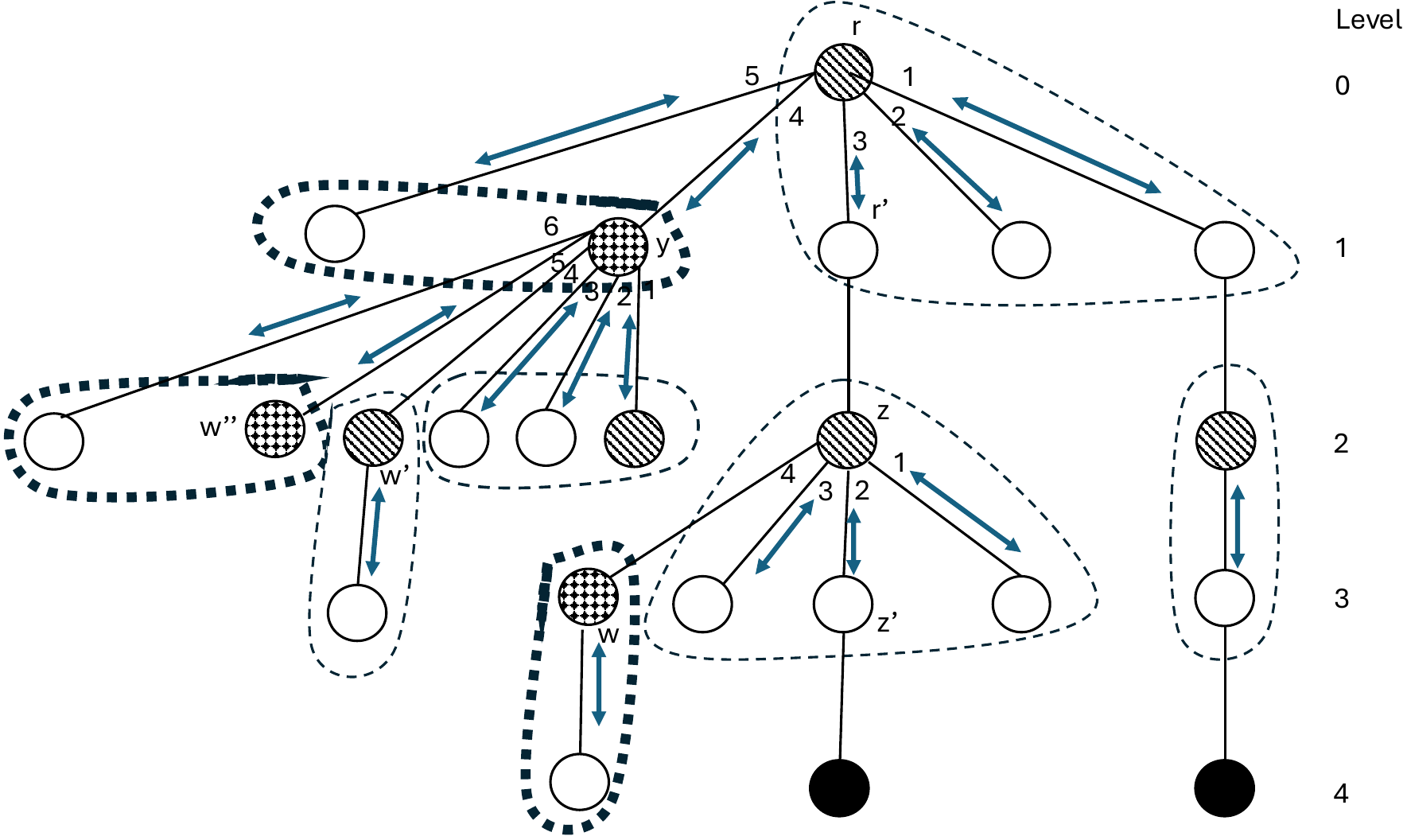}
    \caption{An illustration of how non-oscillating settlers become oscillating when DFS progresses adding more nodes in $T_{DFS}$. Three non-oscillating settlers ($y$, $w$, and $w''$) in Fig.~\ref{fig:oscillation1} now became oscillating (shown in the circle with diamonds) covering empty nodes in their respective groups.}
    \label{fig:oscillation2}
\end{figure}

\begin{figure}[!t]
    \centering
    \includegraphics[width=0.6\linewidth]{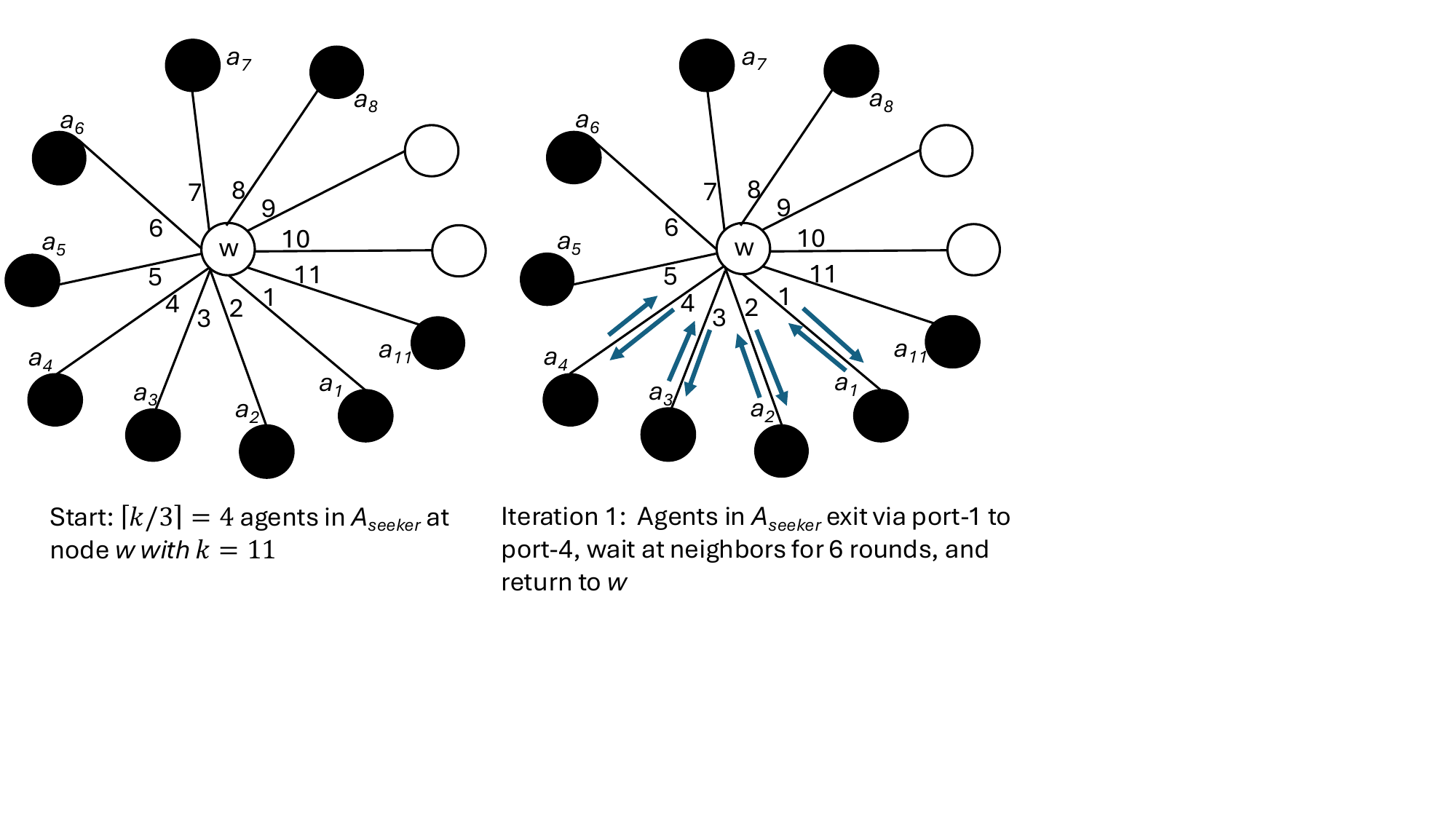}
    \includegraphics[width=0.27\linewidth]{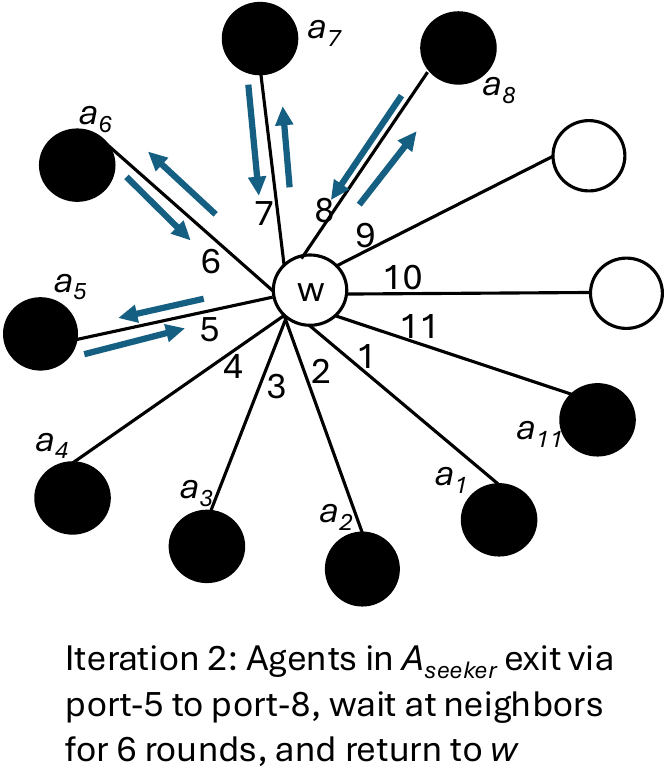}
    \includegraphics[width=0.27\linewidth]{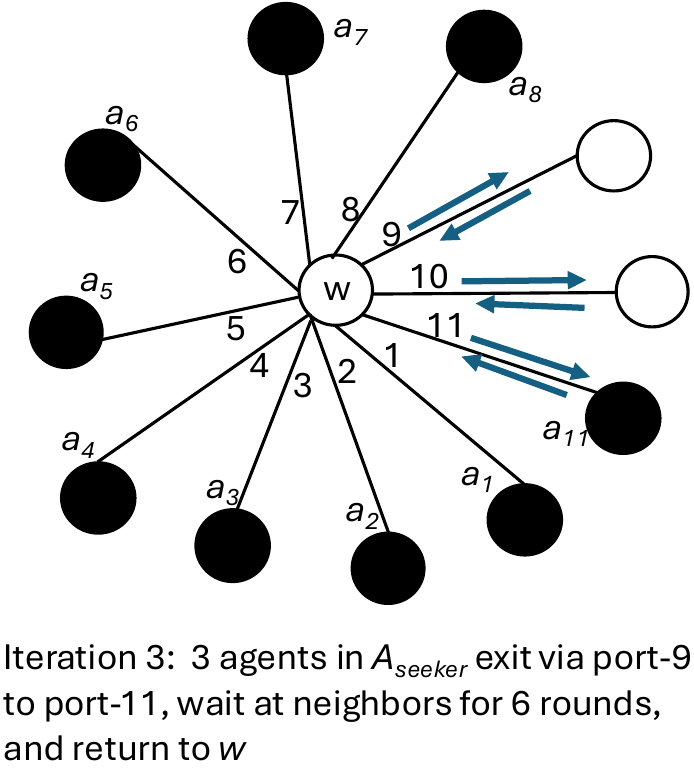}
    \includegraphics[width=0.6\linewidth]{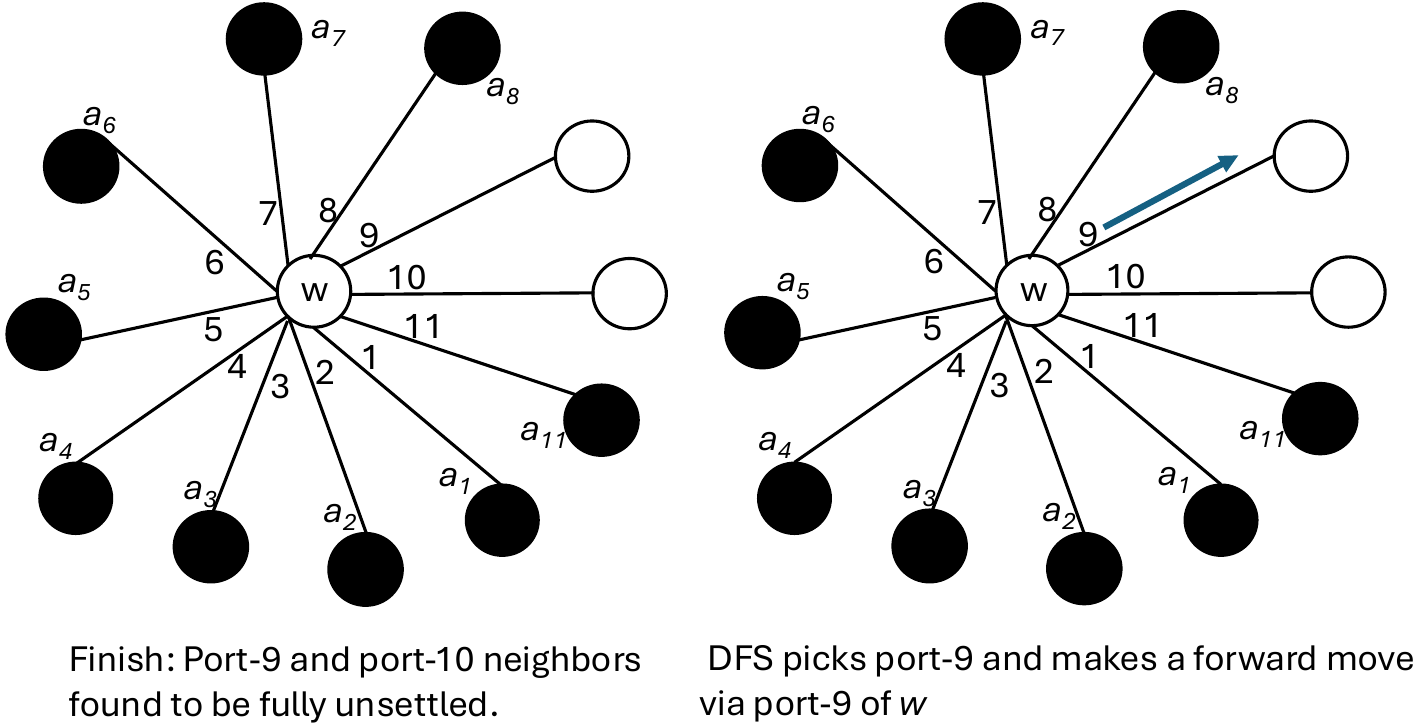}
    \caption{An illustration of how \texttt{Sync\_Probe()} (Algorithm \ref{alg: sync_super_probe()}) finds a fully unsettled neighbor at a node $w$ (if exists one). In the figure, port-9 and port-10 neighbors of $w$ were found to be fully unsettled. DFS does its forward move by exiting through port-9 of $w$.}
    \label{fig:sync-probe}
\end{figure}

\subsection{Empty Node Selection and Oscillation during DFS}
DFS starts from source $s\in V$ where all $k$ agents are initially located. The largest ID agent $a_{max}$ takes the responsibility of running DFS. DFS settles one agent at $s$. 
Regarding other fully unsettled nodes the DFS visits, it settles an agent on each such node that is at even depth in $T_{DFS}$ built by the DFS so far. If no branching node, Lemma \ref{lemma:kover3} immediately satisfies. If there is a branching node, the DFS adjusts while at the child node of that branching node by either putting an extra settler (branching node is at any even depth) or removing a settler (branching node is at odd depth and at least one of its children is a leaf of $T_{DFS}$). Fig.~\ref{fig:oscillation2} illustrates these ideas. Consider the agents at nodes $y,w',w''$. They were non-oscillating in Fig.~\ref{fig:oscillation1}. Suppose DFS made a forward move from node $y$ through port 6. The agent at $w''$ now sees there is one sibling node added in $T_{DFS}$ and there is no other agent to cover it. The agent at $w''$ becomes an oscillating settler to cover that new sibling node. 
Therefore, whether to put an extra settler or remove a settler whether a settler oscillates or not and when it oscillates can be done during DFS. 

\begin{observation}
{\it \texttt{Empty\_Node\_Selection()} can be implemented during DFS even when $T_{DFS}$ is monotonically growing.}
\end{observation}

\begin{algorithm}[t]
\caption{\texttt{Sync\_Probe()}}\label{alg: sync_super_probe()}
\SetKwInOut{Input}{Input}
\SetKwInOut{Output}{Output}

$(\alpha(w).next, \alpha(w).\checked) \gets (\bot, 0)$\; 

\While{$\alpha(w).\checked \neq \delta_w$}
{
    $\{a_1, a_2, \dots, a_x\}\in A_{seeker}$ with $x=|A_{seeker}|$\;
    $\Delta' \gets \min(x, \delta_w - \alpha(w).\checked)$\;
    \For{$j \gets 1$ \KwTo $\Delta'$}
    {
        assign $a_j$ to the port $N(w, j + \alpha(w).\checked)$\;
        $a_j$ leaves $w$ using its port (suppose node reach is $u_j$), waits at $u_j$ for 6 rounds, and returns to $w$\; 
    }
    \If{there exists $a_j$ that did not meet an agent $\alpha(u_i)$}
    {
        $\alpha(w).next \gets j + \alpha(w).\checked$\; 
        
        \textbf{break} the while loop\;
    }
    $\alpha(w).\checked \gets \alpha(w).\checked + \Delta'$\;
}
\end{algorithm}

\subsection{Synchronous Probing}
We present an algorithm \texttt{Sync\_Probe()} (pseudocode in Algorithm \ref{alg: sync_super_probe()}) which finishes probing at a node $w$ in {\sync} in 
$O(1)$ rounds. Fig.~\ref{fig:sync-probe} provides an illustration. The goal in \texttt{Sync\_Probe()} is to verify whether there is at least a neighbor of $w$ that is fully unsettled and return the port of $w$ leading to that node (if there exists one).  \texttt{Sync\_Probe()} uses agents in $A_{seeker}$ present on $w$ to search for a fully unsettled node in $N(w)$. 
We have that $|A_{seeker}|= \lceil k/3\rceil$.
Since there are $k\leq n$ agents, $T_{DFS}$ never has more than $k$ nodes (both empty and non-empty). Since $\geq \lceil k/3\rceil$ agents in $A_{seeker}$  run \texttt{Sync\_Probe()} at $w$, running \texttt{Sync\_Probe()} (at most) three times at $w$ should cover all neighbors in $N(w)$, taking total 6 rounds with 2 roundtrip round per iteration. 
However, since some nodes of $T_{DFS}$ may be empty (Algorithm \ref{algorithm:emptynode}), the empty neighbor in $N(w)$ found by \texttt{Sync\_Probe()} may be the one that was already visited. To avoid this, we ask our $\geq \lceil k/3\rceil$ seeker agents to return only after waiting at the neighboring nodes for 6 rounds. 
While waiting for 6 rounds, if that node was previously visited by DFS, it is guaranteed that a seeker agent will meet an oscillating settler doing its trip at that node in those 6 rounds. 

We implement \texttt{Sync\_Probe()} with a variable $\alpha(w).\checked \in [0, \delta_w]$. The variable $\alpha(w).\checked$ stores the most recently checked port number, initially, $\alpha(w).\checked$ is set to $0$.  $\alpha(w).\checked=l$ implies that the neighbors $N(w,1),N(w,2),\dots,N(w,l)$ are not fully unsettled. Let $x=|A_{seeker}|$. 
In the first iteration of \texttt{Sync\_Probe()}, the $\min\{x,\delta_w\}$ agents visit $\min\{x,\delta_w\}$ neighbors in parallel (one to one mapping), wait at the respective neighbors they reached for 6 rounds, and return to $w$ (Lines 5--7). 
If there is an agent in $A_{seeker}$ that does not find a settler at the neighbor of $w$ it visited, then the node visited by that agent must be fully unsettled. 
In such a case, the port used by one of those agents (picked smallest port neighbor in case of multiple) is stored in $\alpha(w).\nxt$, and the while loop terminates (Line 10). If all agents find a settler (no empty node), we repeat \texttt{Sync\_Probe()} until $\min\{k,\delta_w\}$ ports are checked. 
We have the following guarantees. 
\begin{lemma}
\label{lemma:syncprobe}
 {\it  
  At the end of \texttt{Sync\_Probe()} at a node $w\in V$, it is guaranteed that:
(i) if there exists a fully unsettled node in $N(w)$, then $N(w,\alpha(w).\nxt)$ is unsettled, and 
(ii) if there are no fully unsettled nodes in $N(w)$, then $\alpha(w).\nxt=\bot$ holds true.
\texttt{Sync\_Probe()} at a node $w$ finishes in  $O(1)$ rounds.
}
\end{lemma}

\begin{proof}
    We prove the first two cases by contradiction. Consider first Case (i). Suppose there exists a fully unsettled node $u\in  N(w)$ which $\texttt{Sync\_Probe()}$ classifies as settled. For that to happen, $u$ must have been a node in $T_{DFS}$, and $u$ either has an agent settled or covered through oscillation by the agent at a node previously visited. This is a contradiction since $u$ was not previously visited by the DFS and hence cannot be in $T_{DFS}$. Therefore, $\texttt{Sync\_Probe()}$ cannot classify it as settled since there is no agent settled or covered it through oscillation.  Case (ii) follows from a similar argument.   

We now prove the runtime of $\texttt{Sync\_Prove()}$.
Since there are $k\leq n$ agents, the size of $T_{DFS}$, $|T_{DFS}|=k$.
Consider  
$\texttt{Sync\_Prove()}$ at node $w$. 
Irrespective of degree $\delta_w$ of $w$, there must be a neighbor in $N(w)$ that is fully unsettled found by $\texttt{Sync\_Probe()}$ after probing at most $k-1$ other nodes. Since $|A_{seeker}|= \lceil \frac{k}{3}\rceil$, (at most) $k-1$ neighbors of $w$ are probed in 3 repetitions of $\texttt{Sync\_Probe()}$. Each  $\texttt{Sync\_Probe()}$ finishes in 8 rounds, 6 rounds wait at the neighbor and 2 rounds roundtrip from $w$ to $w$'s neighbor. 
Therefore, $\texttt{Sync\_Probe()}$ at $w$ finishes in $24 = O(1)$ rounds.

\end{proof}

\subsection{Asynchronous Probing}
We present an algorithm \texttt{Async\_Probe()} (pseudocode in Algorithm \ref{algorithm:asyncprobe}) which finishes probing at a node $w$ in {\async} in $O(\log k)$ epochs. Fig.~\ref{fig:async-probe} provides an illustration. 
The goal in \texttt{Async\_Probe()} is to  verify whether there is at least a neighbor of $w$ that is fully unsettled and return the port of $w$ leading to that node (if there exits one).  \texttt{Async\_Probe()} uses agents present at $w$ as well as the settled helper agents from the previously visited non-empty neighbors in $N(w)$ while running \texttt{Async\_Probe()} at $w$ to search for a fully unsettled node in $N(w)$. 
Let $x=|\aset(w)\setminus\{\alpha(w)\}|$, i.e., there are $x$ agents at $w$ when \texttt{Async\_Probe()} is invoked, excluding the agent settled at $\alpha(w)$. We have that $x\geq 1$, otherwise the DFS is complete and dispersion is achieved at $w$.

\begin{algorithm}[t]
\caption{\texttt{Async\_Probe()}}
\label{algorithm:asyncprobe}
\SetKwInOut{Input}{Input}
\SetKwInOut{Output}{Output}
\SetKwFunction{FProbe}{Probe}  

\SetKwFunction{FProbe}{\texttt{Async\_Probe()}}
\SetKwProg{Fn}{Function}{:}{end}

Let $w \gets \nu(a_{\texttt{max}})$\;
$(\alpha(w).next, \alpha(w).\texttt{checked}) \gets (\bot, 0)$\; 
$A_{guest}(w)\leftarrow \emptyset$\;
\While{$\alpha(w).\texttt{checked} \neq \delta_w$}
{
    let $a_1, a_2, \dots, a_x$ be the agents in $A(w) \setminus \{\alpha(w)\}$\;
    $\Delta' \gets \min(x, \delta_w - \alpha(w).\texttt{checked})$\;
    \For{$i \gets 1$ \KwTo $\Delta'$}
    {
        assign $a_i$ to the port $i + \alpha(w).\texttt{checked}$ of $w$ leading to neighboring node $u_i$\;
        $a_i$ makes a round trip: $w \to u_i \to w$\;
        \If{$a_i$ finds a settler at $u_i$}
        {
            it brings the settler $\alpha(u_i)$ to $w$\;
            $A_{guest}(w) \leftarrow A_{guest}(w) \cup \{\alpha(u_i)\}$\;
        }
    }
    \If{$\exists$ $a_i$ that did not bring $\alpha(u_i)$ to $w$}
    {
        $\alpha(w).next \gets i + \alpha(w).\texttt{checked}$\; 
        \textbf{break} the while loop\;
    }
    $\alpha(w).\texttt{checked} \gets \alpha(w).\texttt{checked} + \Delta'$\;
}
\If{$A_{guest}(w)\neq \emptyset$}
{
\texttt{Guest\_See\_Off()}\;
}
\end{algorithm}

\begin{algorithm}[t]
\caption{\texttt{Guest\_See\_Off()}}
\label{algorithm:guestseeoff}
\SetKwInOut{Input}{Input}
\SetKwInOut{Output}{Output}

\SetKwFunction{FProbe}{\texttt{Async\_Probe()}}
\SetKwProg{Fn}{Function}{:}{end}

\While{$A_{guest}(w) \neq \emptyset$}
{
    \If{$|A_{guest}(w)|=1$}
    {
        $\alpha(w)$ and the only agent $u\in A_{guest}(w)$  leave $w$ from the port of $w$ from which $u$ entered $w$ for the first time during  \texttt{Async\_Probe()}\;

        $\alpha(w)$ returns to $w$ after both $\alpha(w)$ and $u$ reached the neighbor node\;
    }
    \Else{
    order the agents $gw_1,gw_2,\ldots,gw_{|A_{guest}(w)|}$ in $A_{guest}(w)$ in an increasing order of their IDs\;
    
    pair the agents $(gw_i,gw_{i+1})$ from the beginning (in total $\lfloor \frac{|A_{guest}(w)|}{2}\rfloor$ pairs with at most one left unpaired)\;

    pair $(gw_i,gw_{i+1})$ leave $w$ from port of $w$ from which $gw_i$ entered $w$ for the first time during \texttt{Async\_Probe()}\;
    $gw_{i+1}$ returns to $w$ after both $gw_i,gw_{i+1}$ reached the neighbor node of $w$\;
    
    \If {$\lfloor \frac{|A_{guest}(w)|}{2}\rfloor$ agents returned $w$}
    {
    $A_{guest}(w)\leftarrow \lfloor \frac{|A_{guest}(w)|}{2}\rfloor$ agents available at $w$ and one possibly unpaired\;  
    }
    }
    }
\end{algorithm}

The $\min\{x,\delta_w\}$ agents visit in parallel $\min\{x,\delta_w\}$ neighbors and then return to $w$.  Each agent brings back the agent $\alpha(u_i)$ settled (if exists) at the neighbor $u_i$ it visited ($\alpha(u_i)$ is the settled helper agent). If no such agent $\alpha(u_i)$ at $u_i$, then that neighbor $u_i$ must be fully unsettled. This property is guaranteed through \texttt{Guest\_See\_Off()} (pseudocode in Algorithm \ref{algorithm:guestseeoff}). If $w$ is the first node where {\texttt{Async\_Probe()} is executing, it is indeed true that an empty neighbor is fully unsettled. If $w$ is not the first node (let the previous node be $w'$, then \texttt{Guest\_See\_Off()} executed at that node guarantees that a neighbor found empty  while {\texttt{Async\_Probe()} is running at $w$ is indeed fully unsettled.  If a neighbor in $N(w)$ is found empty, the port of $w$ leading to that neighbor is stored in $\alpha(w).\nxt$ and {\texttt{Async\_Probe()} is complete. 
Otherwise, if all $x$ agents bring back one agent each (i.e., no visited neighbor is fully unsettled), then there are $2x$ agents on $w$, excluding $\alpha(w)$. In the second iteration of \texttt{Async\_Probe()}, these $2x$ agents visit the next $2x$ neighbors in search of fully unsettled neighbors. As long as no fully unsettled neighbor is discovered, the number of agents on $w$, excluding $\alpha(w)$, doubles with each iteration of \texttt{Async\_Probe()}. \texttt{Async\_Probe()} stops as soon as an empty neighbor is found or all neighbors of $w$ are visited. 

\begin{lemma}
\label{lemma:asyncprobe}
{\it At the end of \texttt{Async\_Probe()} at each node $w$, it is guaranteed that:
(i) if there exists a fully unsettled node in $N(w)$,  $N(w,\alpha(w).\nxt)$ leads to that node,  
(ii) if there are no fully unsettled nodes in $N(w)$, then $\alpha(w).\nxt=\bot$ holds true. 
\texttt{Async\_Probe()} at $w$ finishes in $O(\log k)$ epochs. 
}
\end{lemma}
\begin{proof}
Cases (i) and (ii) follow similarly as in Lemma \ref{lemma:syncprobe}. 
For the runtime, each iteration of \texttt{Async\_Probe()} takes exactly two epochs.
    Since there are at most $\min\{k,\Delta\}$ settled nodes in $N(w)$ and in each subsequent operation double the number of neighbors in $N(w)$ can be probed, after running \texttt{Async\_Probe()} for at most $O(\log \min\{k,\Delta\})=O(\log k)$ times, either a fully unsettled neighbor node of $w$ will be found, or the search will be concluded with no fully unsettled neighbor. 
\end{proof}
\begin{figure}[!t]
    \centering
    \includegraphics[width=0.28\linewidth]{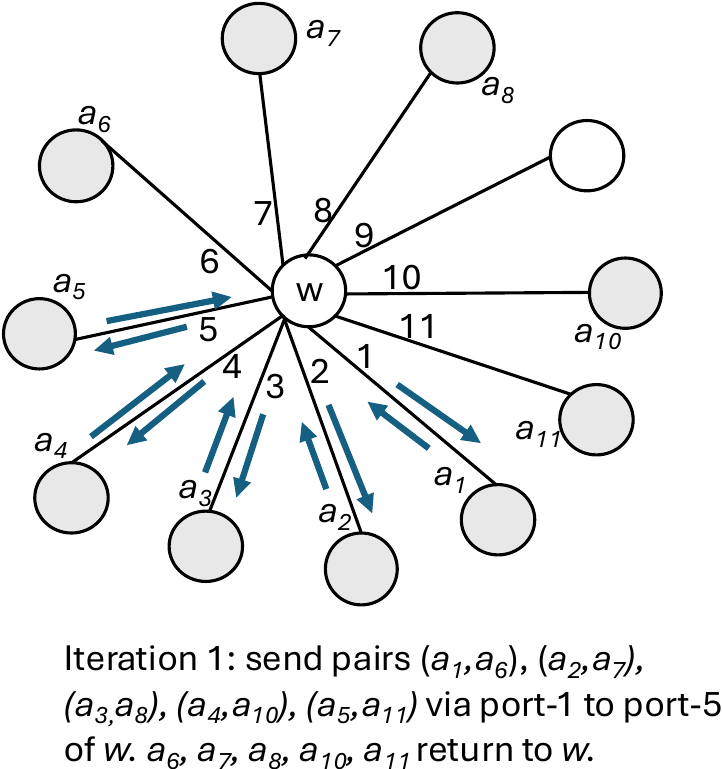}
    \includegraphics[width=0.3\linewidth]{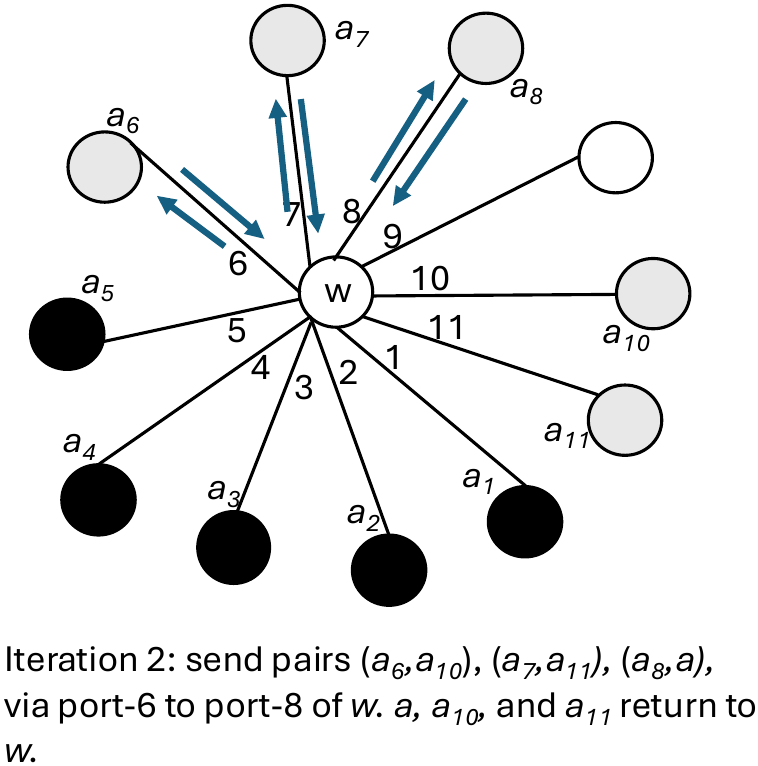}
    \includegraphics[width=0.28\linewidth]{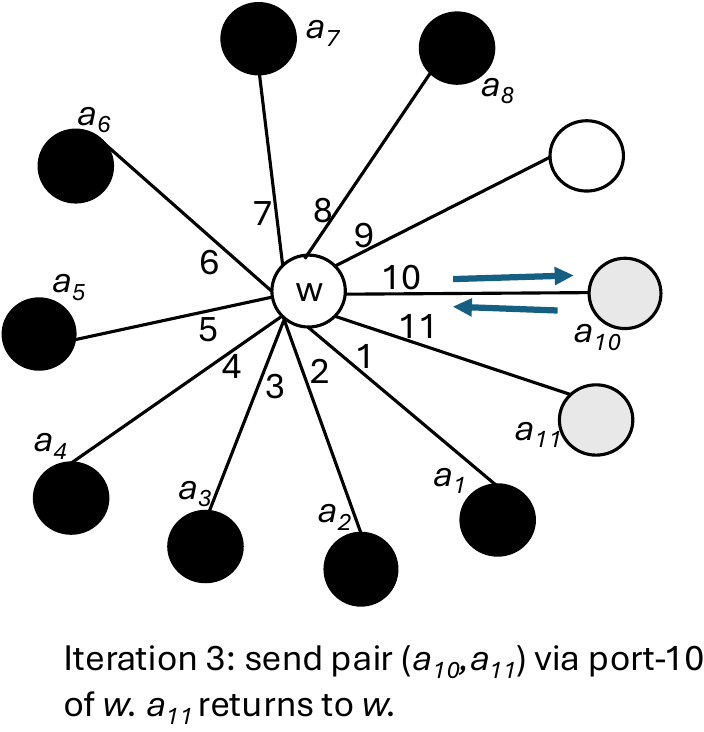}
    \includegraphics[width=0.3\linewidth]{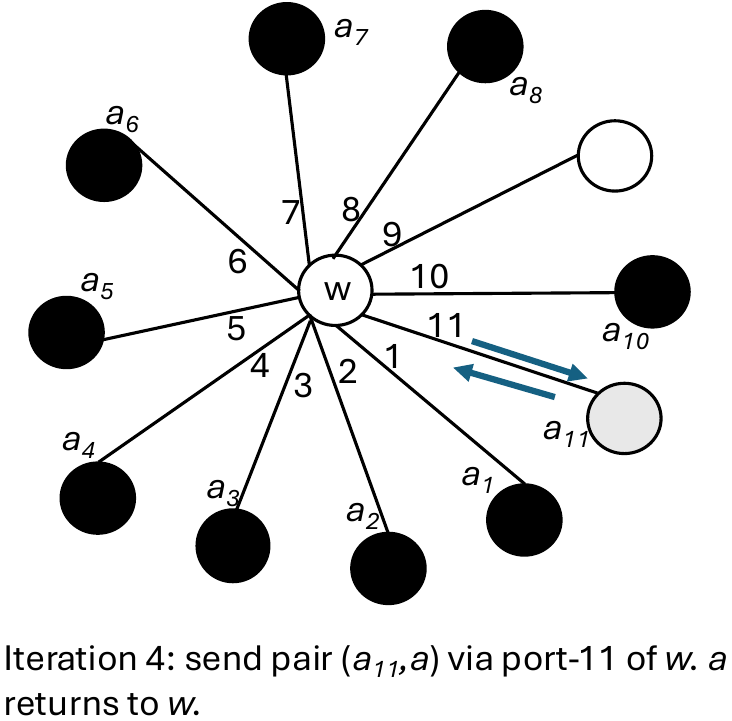}
    \includegraphics[width=0.28\linewidth]{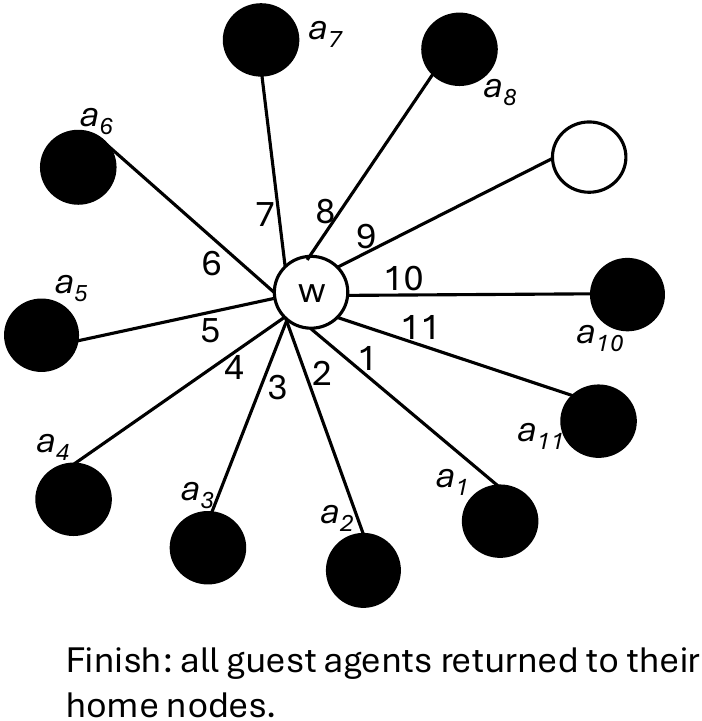}
    \includegraphics[width=0.28\linewidth]{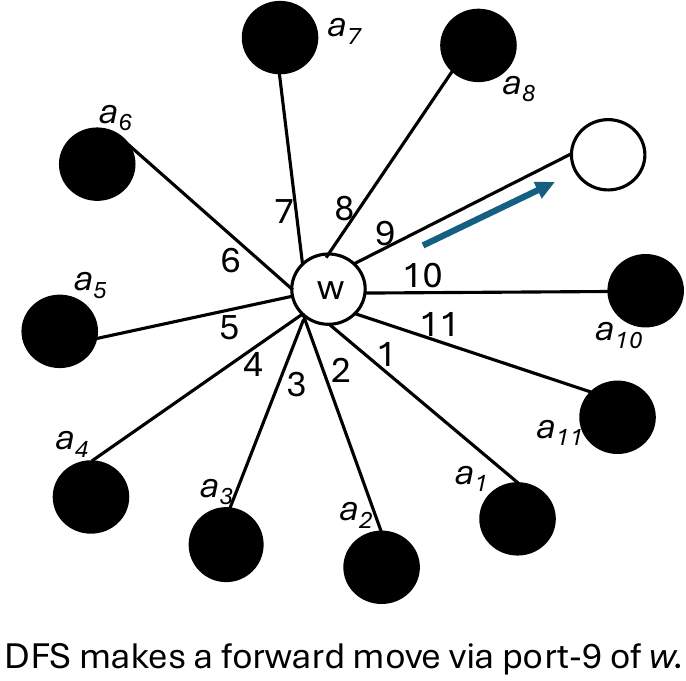}
    \caption{An illustration of how agents brought at node $w$ during \texttt{Async\_Probe()} (Algorithm \ref{algorithm:asyncprobe}) are sent back to their homes executing \texttt{Guest\_See\_Off()} (Algorithm \ref{algorithm:guestseeoff}). DFS does its forward move via exiting through port-9 of $w$ (which was found to be fully unsettled during \texttt{Async\_Probe()}). If no neighbor was found fully unsettled, DFS would go to parent of $w$ via a  backtrack move.}
    \label{fig:guestseeoff}
\end{figure}

After {\texttt{Async\_Probe()} completes at $w$, we need to send back each settled helper agent $\alpha(u_i)$ brought to $w$ to help with probing back to their homes. Only after that the agents in $|\aset(w)\setminus\{\alpha(w)\}|$ can exit $w$. This guarantees that when  {\texttt{Async\_Probe()} executes at the next node, due to asynchrony, the home nodes of the guests at $w$ were not found as empty.  
This sending back of guests to their home nodes is done via \texttt{Guest\_See\_Off()} (Algorithm \ref{algorithm:guestseeoff}). 
Let $A_{guest}(w)$ be the set of such settler helper agents brought to $w$. $|A_{guest}(w)|\leq \min\{k,\Delta\}$.  
For this implementation, 
we pair the agents  in $A_{guest}(w)$. Let $a,b$ be a pair (for odd $|A_{guest}(w)|$, a non-paired agent remains at $w$).  We send $a,b$ to the home of $a$ using the port of $w$ from which $a$ entered $w$ from its home $h(a)$ during {\texttt{Async\_Probe()}}. Each agent in $A_{guest}(w)$ stores such information in its memory. $b$ returns to $w$ after making sure $a$ reached its home $h(a)$. 
After one agent from each pair returns to $w$, we have $|A_{guest}(w)|=|A_{guest}(w)|/2.$ We then again make pairs among the agents in $A_{guest}(w)$ and settle half of the agents to their homes.
This process continues until either $A_{guest}(w)=\emptyset$ (even) or $|A_{guest}(w)|=1$ (odd). For the case of $|A_{guest}(w)|=1$, $\alpha(w)$ is used as a pair to see off the only agent in $A_{guest}(w)$ to its home.  This careful implementation of settling the guests to their homes plays a crucial role in guaranteeing that the empty neighbors found in {\texttt{Async\_Probe()}} are in fact the fully unsettled empty neighbors. 

\begin{lemma}
\label{lemma:guestseeoff}
{\it At the end of \texttt{Guest\_See\_Off()} at node $w$, each agent brought to $w$ goes backs to its home node and
\texttt{Guest\_See\_Off()} at $w$ finishes in $O(\log k)$ epochs. 
}
\end{lemma}
\begin{proof}
Let $u_i$ be a neighboring node of $w$. Agent $\alpha(u_i)$ while brought to $w$ (the first time coming to $w$) remembers the port of $w$ from which it entered $w$. When it is sent back during  {\texttt{Guest\_See\_Off()}}, it uses the same port to exit $w$. And hence, $\alpha(u_i)$ reaches $u_i$. 

Regarding runtime, we have that $|A_{guest}(w)|\leq \min\{k,\Delta\}$. Each iteration of {\texttt{Guest\_See\_Off()}} settles $\lfloor \frac{|A_{guest}(w)|}{2}\rfloor$ guests at node $w$ to their home nodes.  Therefore, executing {\texttt{Guest\_See\_Off()}} for $\lceil \log |A_{guest}(w)| \rceil +1$ iterations settles all guests in  $A_{guest}(w)$ to their home nodes.  It is easy to see that each iteration of {\texttt{Guest\_See\_Off()}} finishes in 2 epochs. Therefore, the total time is $O(\log \min\{k,\Delta\})=O(\log k)$ epochs. 
\end{proof}

\begin{figure}[!t]
    \centering
    \includegraphics[width=0.26\linewidth]{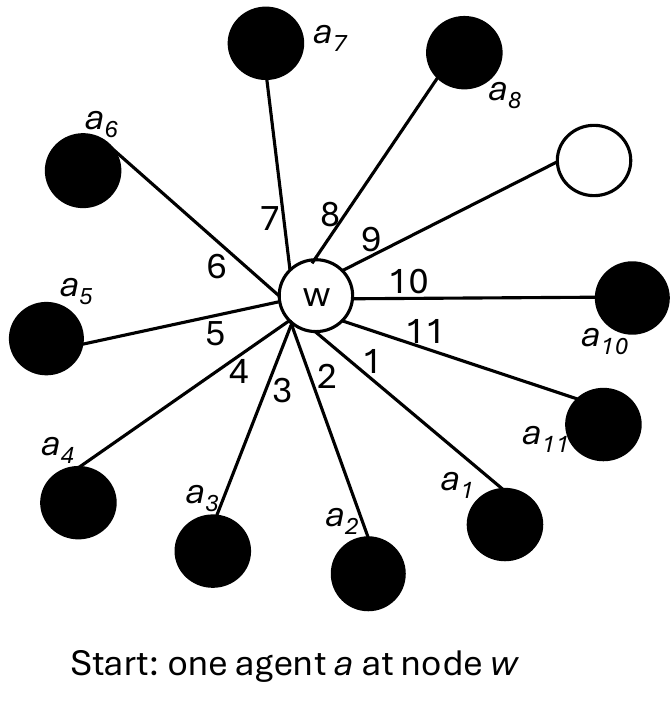}
    \includegraphics[width=0.3\linewidth]{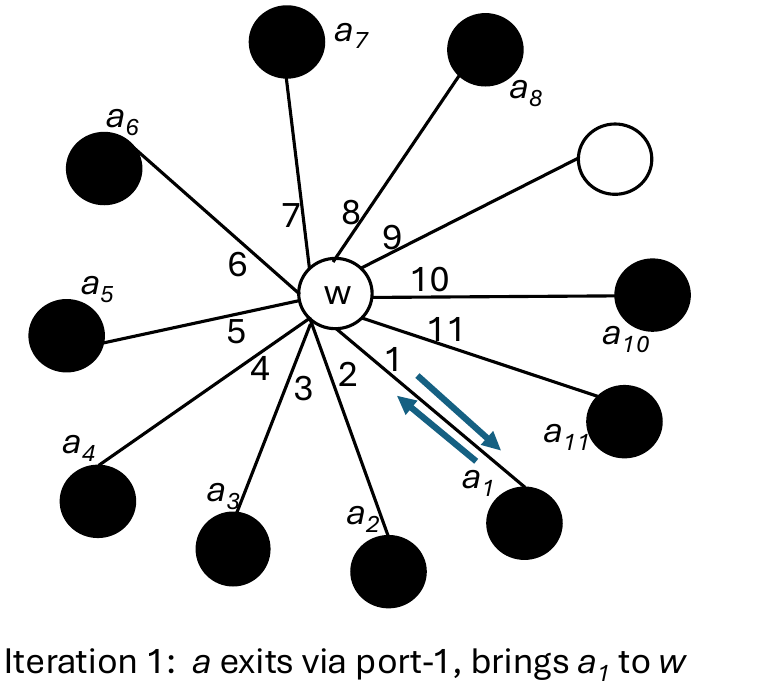}
    \includegraphics[width=0.25\linewidth]{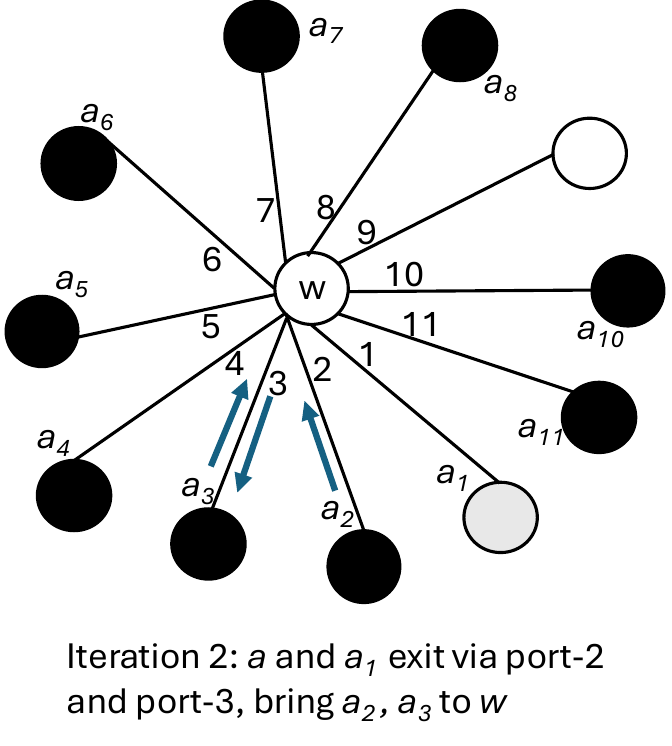}
    \includegraphics[width=0.3\linewidth]{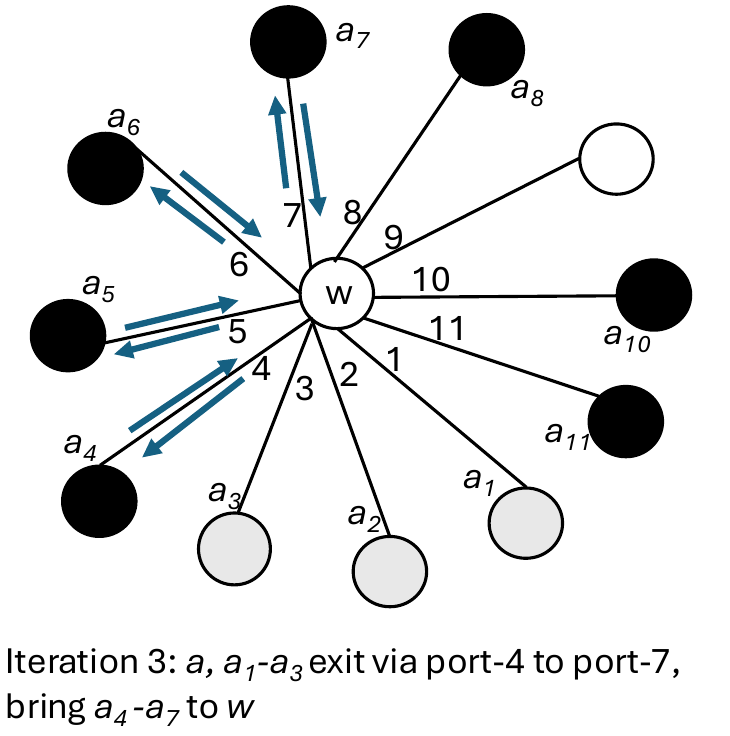}
    \includegraphics[width=0.25\linewidth]{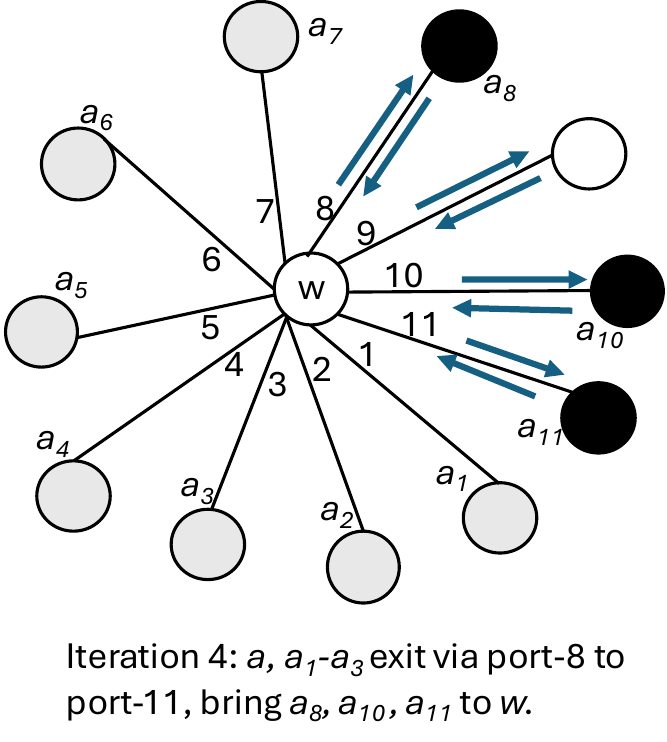}
    \includegraphics[width=0.3\linewidth]{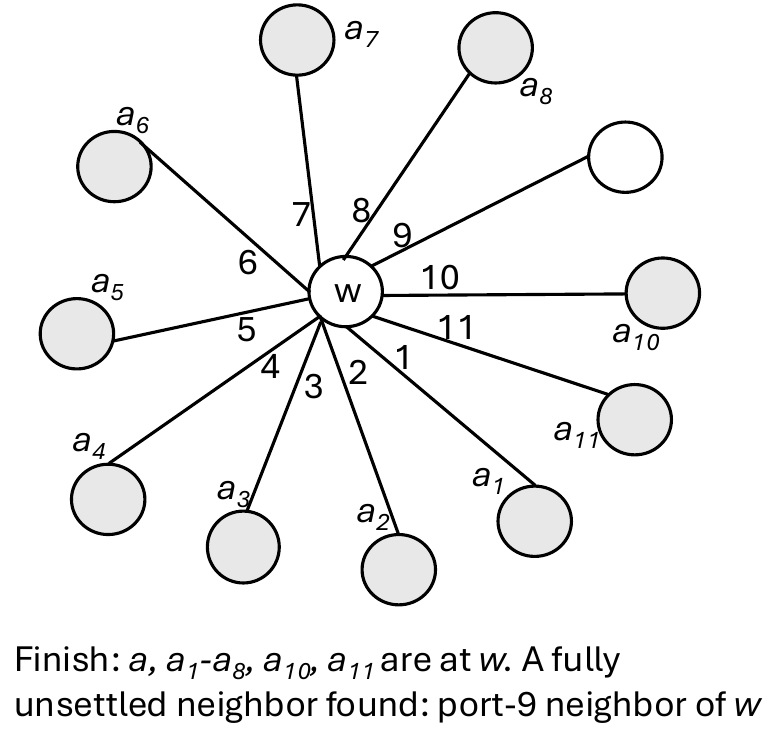}
    \caption{An illustration of how \texttt{Async\_Probe()} (Algorithm \ref{algorithm:asyncprobe}) finds a fully unsettled neighbor at a node $w$ (if exists one). In the figure, port-9 neighbor of $w$ is found to be fully unsettled. If port-9 neighbor was occupied by an agent $a_9$, \texttt{Async\_Probe()} would report non-existence of such node.}
    \label{fig:async-probe}
\end{figure}

\section{\texorpdfstring{\sync}{} Rooted  Dispersion Algorithm}\label{section:RootedSync}
In this section, we present our first algorithm, $\rootsync$, that solves the dispersion of $k\leq n$ agents initially located at a single node $s\in V$ in $O(k)$ rounds with $O(\log(k+\Delta))$ bits per agent in {\sync}. The algorithm is DFS-based which constructs the DFS tree, $T_{DFS}$. 



\begin{algorithm}[t]
\caption{$\rootsync$}\label{alg: sync_rooted}
\SetKwInOut{Input}{Input}
\SetKwInOut{Output}{Output}


$a_{\min}.settled \gets \top$ and $a_{\min}.parent \leftarrow \bot$ at node $s$\; 
$A_{seeker}\leftarrow$ the $\lceil \frac{k}{3}\rceil$ agents in $A$ with largest IDs, except $a_{\max}$\;
$A_{explorer}(s)\leftarrow A(s)\backslash (A_{seeker}\cup\{a_{min}\})$\;
agents in $A_{explorer}(s)\cup A_{seeker}$ exit $s$ via port-1 neighbor\; 





\While{$k$ different nodes visited}
{
    $w \gets$ the node $\alpha(a_{\texttt{max}})$ where $a_{max}$ is positioned\;
    \texttt{Sync\_Probe()} at $w$\;
    
    \If{$\alpha(w).next \neq \bot$}
    {
        \texttt{Forward\_Move()\;}
    }
    \Else
    {   
        \texttt{Backtrack\_Move()\;}
    }
}
$A_{explorer}^{un} \leftarrow$ agents in $A_{explorer}(w)$ not yet settled\;
agents in $A_{explorer}^{un}\cup A_{seeker}$ (and others not yet settled) follow $parent$ pointers at each node until reaching $s$, the root of $T_{DFS}$\;

agents in $A_{explorer}^{un}\cup A_{seeker}$ (and others not yet settled) re-traverse $T_{DFS}$ through using sibling pointers and settle at empty nodes\;
\label{line: linear_DFS}
\end{algorithm}

\begin{algorithm}[t]
\caption{\texttt{Forward\_Move()}}\label{alg: rooted_sync_forward}
\SetKwInOut{Input}{Input}
\SetKwInOut{Output}{Output}

\If{$\alpha(w).firstchild = \perp$}
{
    
    $\alpha(w).firstchild \leftarrow  \alpha(w).next$\; 

    $\alpha(w).latestchild \leftarrow \alpha(w).next$\; 
}
\Else
{
    $a_{\texttt{max}}$ moves to $\alpha(w).latestchild$, sets  $\alpha(w).latestchild.nextsibling\leftarrow \alpha(w).next$, and 
    %
     returns to $w$\; 
}
$A(w)\leftarrow$ agents present at $w$, except the settler $\alpha(w)$, if any\;
all agents in $A_{explorer}(w)\cup A_{seeker}$ exit $w$ through port $N(w, \alpha(w).next)$ to reach node $u$; $u$ must be a fully unsettled node\;

\If{$u$ is at odd depth} 
{
 
    $a_{min}\leftarrow$ agent with smallest ID in $A(u)$\; 
    $x\leftarrow$ the number of times \texttt{Sync\_Probe()} found $\alpha(w).next\neq \bot$ at $w$ so far\; 
    \If{$x \leq 3$}
    {ask $\alpha(w)$ to include $u$ in its oscillation trip\;}
    \Else
    {\If{$x \mod 3=1$} 
      {
      $a_{\min}.settled \gets \top$, $a_{\min}.parent \gets a_{\texttt{max}}.\texttt{pin}$ (which leads to $w$)\;
      }
      \Else 
      {
       \If{$(x-1) \mod 3=1$}
       {
       $\alpha(u')\leftarrow$ the agent settled at node $u'$ reached through $(x-1)$-th  $\alpha(w).next$ port at $w$\;}
       \If{$(x-1) \mod 3=2$}
       {
       $\alpha(u')\leftarrow$ the agent settled at node $u'$ reached through $(x-2)$-th  $\alpha(w).next$ port at $w$\;
       }
    ask (sibling settler) $\alpha(u')$ to include $u$ in its oscillation trip\;}
       
      }
    }

\Else
{
    $a_{min}\leftarrow$ agent with smallest ID in $A(u)$\; 
    $a_{\min}.settled \gets \top$,
    $a_{\min}.parent \gets a_{\texttt{max}}.\texttt{pin}$\;
}
$\alpha(u).\mathtt{treelabel} \leftarrow a_{\max}.\mathtt{treelabel}$\; 
\end{algorithm}

\begin{algorithm}[t]
\caption{\texttt{Backtrack\_Move()}}\label{alg: rooted_sync_backward}
\SetKwInOut{Input}{Input}
\SetKwInOut{Output}{Output}
\SetKwFunction{FProbe}{Probe}  

\SetKwFunction{FProbe}{Probe()}
\SetKwProg{Fn}{Function}{:}{end}

\If{$w$  is at even depth}
{
$p_w\leftarrow$ parent of $w$ which is reached via $a_{\texttt{max}}.\texttt{pin}$\; 
let $w$ be the $x$-th leaf children of $\alpha(p_w)$ in the DFS tree\; 
\If{$x>1$}
{
\If{$x\mod 3 = 0$}
{
$\alpha(u')\leftarrow$ the agent settled at node $u'$ which is the $(x-2)$-th leaf children of $\alpha(p_w)$\; 
}
\If{$x\mod 3 = 2$}
{
$\alpha(u')\leftarrow$ the agent settled at node $u'$ which is the $(x-1)$-th leaf children of $\alpha(p_w)$\; 
}
remove agent $\alpha(w)$ settled at $w$ setting $\alpha(w).settled\leftarrow \bot$\;
ask (sibling settler) $\alpha(u')$ to include node $w$ in its oscillation trip\;
}
}
    agents in $A_{explorer}(w) \cup A_{seeker} \backslash \{\alpha(w)\}$ leave $w$ via port $a_{\max}.\texttt{pin}$ from $w$\; 

\end{algorithm}


In $\rootsync$ (Algorithm~\ref{alg: sync_rooted}), the largest ID agent, denoted as $\amax$, serves as the leader and 
conducts  DFS. 
The non-leader agents move with the leader and one of them settles at a fully unsettled node when they visit it. If $\amax$ encounters a fully unsettled node alone, it settles itself on that node, achieving dispersion. During DFS, $\amax$ currently at node $w$ must determine 
\begin{itemize}
    \item 
[(i)] whether there is a fully unsettled neighbor node of $w$, and 
\item [(ii)] if so, which neighbor of $w$ is fully unsettled.  
\end{itemize}
%
%
%
%
$\amax$ makes this decision through function \texttt{Sync\_Probe()} (Algorithm \ref{alg: sync_super_probe()}).
$a_{\max}$ uses $ \lceil \frac{k}{3}\rceil$ seeker agents to execute \texttt{Sync\_Probe()} in $O(1)$ rounds. 

$\rootsync$ runs DFS either in {\em forward} phase or {\em backtrack} phase. The forward phase is handled by \texttt{Forward\_Move()}~(Algorithm \ref{alg: rooted_sync_forward}) and the backtrack phase is handled by \texttt{Backtrack\_Move()}~(Algorithm \ref{alg: rooted_sync_backward}). 
The forward phase is executed at $w$, if \texttt{Sync\_Probe()}
at $w$ finds at least one neighbor of $w$ fully unsettled, otherwise the backtrack phase.  After DFS completes (i.e., $T_{DFS}$ has $k$ nodes), the seeker agents then go to the root of $T_{DFS}$ and then re-traverse $T_{DFS}$ to settle at its empty nodes. After all this is done, the $k$ nodes of $T_{DFS}$ finally have each node occupied and dispersion is achieved.




We now provide details.
In $\rootsync$, every agent $a$ maintains a variable $a.\settled \in \{\bot, \top\}$. 
Let $\alpha(w)$ be the settler/oscillating agent at $w$.
Agent $\alpha(w)$ maintains the following variables, $\alpha(w).\parent$ (the parent of $w$), $\alpha(w).firstchild$ (the first child of $w$), $\alpha(w).sibling1$, $\alpha(w).sibling2$ (two siblings of a child of $w$), $\alpha(w).latestchild$ (the latest child of $w$ to facilitate sibling traversal), $\alpha(w).depth$ (the depth of $w$ in $T_{DFS}$), and 
$\alpha(w).\nxt$ $\in [1, \delta_w] \cup \{\bot\}$.
In  \texttt{Sync\_Probe()} at node $w$,  if there exists a fully unsettled node  in $N(w)$,
the corresponding port number will be in $\alpha(w).\nxt$.
More precisely, an integer $i$ such that $N(w, i) = u$ and $\alpha(u) = \bot$
is assigned to $\alpha(w).\nxt$.
If all neighbors $N(w)$ are settled, 
$\alpha(w).\nxt$ will be set to $\bot$.

At round 0, all agents are located at node $s$. Let $A_{seeker}\subset A$ be the $\lceil k/3\rceil$ agents at $s$ with largest IDs, except $a_{\max}$. 
The agents in $A_{seeker}$ will run $\texttt{Sync\_Probe()}$ at each node. 
Let $A_{explorer}(w)\in A\backslash A_{seeker}$ denote the agents on any node $w$ except $\alpha(w)$  settled at $w$.
Let $a_{\min}(w)$ be the smallest ID agent at $w$ among the agents in $A_{explorer}(w)$. 
Agent $a_{\min}(s)$ settles at node $s$ and sets $a_{\min}(s).\parent\leftarrow \bot$. 

\vspace{2mm}
\noindent{\bf Forward move.}
The DFS starts at $s$ in the forward phase.
As long as there is at least a fully unsettled node in $N(s)$ (for multiple such nodes, pick one associated with the smallest port), all agents in $A_{explorer}(s)\cup A_{seeker}$  move to one of those nodes together (Line 7 of \texttt{Forward\_Move()}). We call this kind of movement \emph{forward move}.  

%
Suppose a forward move from node $w$ brought  agents $A_{explorer}(u)\cup A_{seeker}$ to $u$ ($w$ is the parent of $u$). The following happens at $u$:
\begin{itemize}
\item [1.] Suppose $u$ is the even depth node in $T_{DFS}$. $a_{\min}(u)\in A_{explorer}(u)$ settles at $u$ and sets $a_{\min}(s).settled \gets \top$ and
    $a_{\min}(s).parent \gets a_{\texttt{max}}.\texttt{pin}$.
\item [2.] Suppose $u$ is the odd depth node in $T_{DFS}$. Whether an agent settles at $u$ or not depends on how many times \texttt{Forward\_Move()} has been successful so far from $w$, $u$'s parent. Let $x$ be that number. If $x\leq 3$, we leave $u$ empty and
ask $\alpha(w)$ to oscillate to $u$ (i.e., $\alpha(w)$ will have at most three children nodes to oscillate to). 
If $x\geq 4$, we settle an agent at $u$ if $(x\mod 3)=1$. If $(x-1) \mod 3\in [1,2]$, we leave $u$ empty. The agent settled at  $(x \mod 3=1)$-th children of $w$ oscillates to $x+1$-th and $x+2$-th children of $w$ (i.e., an agent at a sibling node covers two other empty sibling nodes).     
\end{itemize}

\vspace{1mm}
\noindent{\bf Backtrack move.}
When the current location $w$ has no fully unsettled neighbor, all explorers $A_{explorer}(w)\cup A_{seeker}$ exit $w$ through $\alpha(w).parent$ (Line 11 of \texttt{Backtrack\_Move()}). 
We call this kind of movement \emph{backtrack move}.
Suppose there is an agent $\alpha(w)$ settled at node $w$. 
Before exiting $w$, depending on the situation, $a_{\max}$ needs to decide on whether to leave $\alpha(w)$ settled or make it an explorer (i.e.,  add to $A_{explorer}(w)$).
Let $p_w$ be the parent node of $\alpha(w)$. Let $w$ be the $x$-th leaf children of $\alpha(p_w)$ in $T_{DFS}$ known so far. The leaf children means $\alpha(w)$ must be the agent positioned on the leaf of $T_{DFS}$ (if $\texttt{Sync\_Probe()}$ does not return $\alpha(w).next\neq \bot$ even once at $w$, $\alpha(w)$ must be the leaf) and $\texttt{Backtrack\_Move()}$ is needed. 
If $x=1$, we leave $\alpha(w)$ settled.
If $x>1$, we leave $\alpha(w)$ settled only when $(x\mod 3)=1$ (i.e., $x=4,7,10$, etc.) 
In all other values of $x$, we make $\alpha(w)$ the explorer. The agent settled at  $(x \mod 3=1)$-th children of $w$ oscillates to $(x+1)$-th and $(x+2)$-th children of $w$ (i.e., an agent at a sibling node covers two other empty sibling nodes). This transitioning of a settled agent to an unsettled one is to make sure that we never need agents from $A_{seeker}$ to settle while running DFS. 

\begin{lemma}\label{lemma:seeker}
{\it For any $k\geq 3$, during DFS, only (at most) $\lfloor \frac{2k}{3}\rfloor$ agents settle. 
}
\end{lemma}
\begin{proof}
We prove that $\rootsync$ satisfies Lemma \ref{lemma:kover3}. 
Let $\rootsync$ is currently at node $w$. Let $T_{DFS}$ built so far $T_{DFS}$ has size  $|T_{DFS}|=k'< k$. We show that out of $k'$ nodes in $T_{DFS}$, only at most $\lfloor \frac{2k'}{3}\rfloor$ nodes have a settled agent.  Consider $T_{DFS}$ as the arbitrary tree $T$ used in the proof of Lemma \ref{lemma:kover3}. 
It is sufficient to show that the nodes with a settled agent in $T_{DFS}$ are the nodes where agents are settled in $T$. 
The forward phase of $\rootsync$ settles agents on the nodes of $T_{DFS}$ at even depth $0,2,4,\ldots$. 
Suppose no node in $T_{DFS}$ is a branching node. Then no node at an odd depth of $T_{DFS}$ has a settled agent and $\rootsync$ satisfies Lemma \ref{lemma:kover3}. Suppose there is at least a branching node $w$ in $T_{DFS}$. The branching node $w$ is either at odd depth (empty) or even depth (occupied). 
Let $w$ be at an odd depth (empty). 
Let $C_w=\{c_{w1},c_{w2},\ldots,c_{w\nu}\}$ be the $\nu$ children of $w$. Whenever DFS reaches child $c_{wi}$ from $w$ in the forward phase, it settles an agent at $c_{wi}$. If the forward phase continues from $c_{wi}$, there is at least a child of $c_{wi}$ that will be empty (since it will be at an odd depth). If the forward phase cannot proceed from $c_{wi}$, it must be the case that $c_{wi}$ is a leaf in $T_{DFS}$. Suppose there are $\nu\geq x>1$ such children of $w$ so far which are leaves in $T_{DFS}$. Since $w$ is empty and all $x$ children have $w$ as their parent, $\rootsync$ removes 2 agents out of every 3 children of $w$. The removal of the settled agent at a child of $w$ happens in the backtrack phase following the forward phase to that child. Therefore, it satisfies the case of $w$ being a parent of the leaf in Lemma \ref{lemma:kover3}. 
Now suppose $w$ is at an even depth (occupied). All $\nu$ children of $w$ will be empty. For $\nu> 4$, one agent is placed on $c_{w4},c_{w7}, \ldots$ to cover 2 children in between the settlers; agent at $w$ covers $c_{w1},c_{w2},c_{w3}$. This again satisfies the case of $w$ being a parent at even depth in Lemma \ref{lemma:kover3}. Since Lemma \ref{lemma:kover3} shows no more than $\lfloor \frac{2k'}{3}\rfloor$ nodes in $T$ have a settled agent, the lemma follows.
\end{proof}

\begin{lemma}
\label{lemma:rootal}
{\it $\rootsync$ 
solves dispersion in $O(k)$ rounds using $O(\Delta \log (k+\Delta))$ bits per agent.}
\end{lemma}

\begin{proof}
Let the DFS be currently at node $w$. As long as there is (at least) a fully unsettled neighbor node of $w$, $\amax$ makes a forward move to that node. If there is no such neighbor, $\amax$ makes a backward move to the parent node of $w$.  Since $G$ is a connected graph, the DFS visits $k$ nodes with exactly $k-1$ forward moves and at most $k-1$ backward moves.  Thus, the number of calls to \texttt{Sync\_Probe()} is at most $2(k-1)$. From Lemma \ref{lemma:syncprobe}, \texttt{Sync\_Probe()} at a node finishes in $O(1)$ rounds. Additionally, the agents in $A_{explorer}^{un}\cup A_{seeker}$ need $k-1$ rounds to reach the root node of $T_{DFS}$ following parent pointers. 
The agents in $A_{explorer}^{un}\cup A_{seeker}$ reach to empty nodes of $T_{DFS}$ and settle, re-traversing $T_{DFS}$, in $O(k)$ rounds following the information about all children of $w$. Therefore, $\rootsync$ finishes in $O(k)$ rounds.
Regarding memory, an agent at node $w$ handles several $O(\log \Delta)$-bit variables, $firstchild$, $\nxt$, $\checked$, $\parent$, $latestchild$, $sibling1$, $sibling2$, and the information about all the ports of $w$ that lead to children of $w$ in $T_{DFS}$ (at most $\Delta$ such ports with each port needing  $O(\log \Delta)$ bits). 
Other variables can be stored in $O(1)$ bits. 
An agent needs $\lceil \log k\rceil $ bits to store its ID.  
Therefore, the memory becomes $O(\Delta\log (k+\Delta))$ bits.
\end{proof}

\noindent{\bf Memory-efficient DFS tree re-traversal.}
We now discuss how the memory complexity of  
$O(\Delta\log (k+\Delta))$ bits per agent in Lemma \ref{lemma:rootal} can be reduced to $O(\log (k+\Delta))$ bits per agent.  Notice that this is due to the information to keep at the agent at node $w$ in Line 14 of Algorithm \ref{alg: sync_rooted} about all the ports of $w$ that lead to children of $w$ in $T_{DFS}$ (at most $\Delta$ such ports with each port needing  $O(\log \Delta)$ bits). This information is used by 
agents in $A_{explorer}^{un}\cup A_{seeker}$ to settle at empty nodes of $T_{DFS}$. 
We describe a technique of {\em sibling pointer} that allows to re-traverse $T_{DFS}$ in $O(\log \max\{k,\Delta\})$ bits, reducing the memory in Lemma \ref{lemma:rootal} to $O(\log \max\{k,\Delta\})$ bits.
Let $C_w=\{c_{w1},c_{w2},\ldots,c_{w\nu}\}$ be the $\nu$ children of $w$ in $T_{DFS}$. 
Consider $i$-th child $c_{w,i}$ of $w$.
\begin{itemize}
\item {\bf Suppose $c_{w,i}$ is at odd depth:} For $i\leq 3$, $\alpha(w)$ stores the information about $c_{w,1},c_{w,2},c_{w,3}$, so that re-traversal visits them in order, and additionally it stores the first sibling pointer to $c_{w,4}$. Notice that $c_{w,4}$ has a settler. We store at $c_{w,4}$ the information about $c_{w,5}$ and $c_{w,6}$, and the next sibling  pointer to $c_{w,7}$.
Therefore, after DFS finishes re-traversing $c_{w,1}$ to $c_{w,3}$, it visits $c_{w,4}$ to collect information about $c_{w,5}$ and $c_{w,6}$ to visit and the next sibling pointer  information to $c_{w,7}$. Since $c_{w,i},i=4,7,\ldots$ have an agent positioned, $\alpha(w)$ can have information about (at most) four children at a time to re-traverse $T_{DFS}$ without needing to store information about all $\nu$ children.  

\item {\bf Suppose $c_{w,i}$ is at even depth:}
We only consider the subset of the children in $C_{wl}\subseteq C_w = \{c_{wl,1},c_{wl,2},\ldots,c_{wl,\nu'}\}$ which are the leaves in $T_{DFS}$.  
$\alpha(w)$ stores the information about $c_{wl,1},c_{wl,2},c_{wl,3}$, so that re-traversal visits them in order, and additionally it stores the first sibling pointer to $c_{wl,4}$. Notice that $c_{wl,4}$ has a settler. We store at $c_{wl,4}$ the information about $c_{wl,5}$ and $c_{wl,6}$, and the next sibling  pointer to $c_{wl,7}$. Therefore, the children of $\alpha(w)$ can be re-traversed keeping information about (at most) four children at a time.
\end{itemize}

\begin{lemma}\label{lemma:retraversal}
{\it After DFS  completes,  re-traversal of $T_{DFS}$ by agents in  $A_{explorer}^{un}\cup A_{seeker}$ starting from root in $O(k)$ time can be done with only $O(\log(k+\Delta))$ bits per agent.}  
\end{lemma}
\begin{proof}
    We have that Algorithm \ref{alg: sync_rooted}, except Line~\ref{line: linear_DFS}, needs only $O(\log(k+\Delta))$ bits and finishes in $O(k)$ rounds. For Line~\ref{line: linear_DFS}, the sibling pointer idea asks to store information about 4 of the children of a node (at any depth of $T_{DFS}$) to re-traverse $T_{DFS}$ by agents in $A_{explorer}^{un}\cup A_{seeker}$. This storage adds  $O(\log \Delta)$ bits/agent.   Therefore, the total memory complexity becomes $O(\log (k+\Delta))$ bits.
\end{proof}


We have the following main theorem combining Lemmas \ref{lemma:rootal} and \ref{lemma:retraversal}.

\begin{theorem}\label{thm: root_sym}
{\it     Algorithm $\rootsync$ solves dispersion within
$O(k)$ rounds using $O(\log(k + \Delta))$ bits per agent in {\sync}.}
\end{theorem}

\section{\texorpdfstring{\async}{} Rooted Dispersion Algorithm}
\label{section:RootedAsync}

In this section, we present our second algorithm, $\rootasync$, that solves the dispersion of $k\leq n$ agents initially located at a single node $s\in V$ in $O(k\log k)$ epochs with $O(\log (k+\Delta))$ bits per agent in {\async}. This algorithm is also DFS-based and extends the technique of \cite{sudo24} [DISC'24] developed for solving dispersion in {\sync}. We use our algorithm \texttt{Async\_Probe()} (Algorithm \ref{algorithm:asyncprobe}) so that DFS at a node $w$ finds a fully unsettled neighbor node of $w$ in $O(\log k)$ epochs in {\async}. Sudo {\it et al.} \cite{sudo24} [DISC'24] provided the $O(\log k)$-round solution just for {\sync}.

The pseudocode for $\rootasync$ is shown in Algorithm \ref{alg: async_rooted_memory}. $\rootasync$ calls  \texttt{Async\_Probe()} (Algorithm \ref{algorithm:asyncprobe}). \texttt{Async\_Probe()} finds a fully unsettled empty neighbor at any node $w$ (if exists) and then  calls  \texttt{Guest\_See\_Off()} (Algorithm \ref{algorithm:guestseeoff}) to send the settled agents used in probing back to their home nodes.  Every agent $a$ maintains a variable $a.\settled \in \{\bot, \top\}$, which decides whether $a$ is an explorer or a settler. In addition, the settler $\alpha(w)$ at node $w$ maintains two variables, $\alpha(w).\parent, \alpha(w).\nxt \in [1, \delta_w ] \cup {\bot}$. 

Again, $a_{\max}$ is the leader that runs DFS; non-leaders follow $a_{\max}$ until they settle. The following is guaranteed each time $\amax$ invokes \texttt{Async\_Probe()} at node $w$:
\begin{itemize}
\item If there exists a fully unsettled neighbor node in $N(w)$, \texttt{Async\_Probe()} finds it and stores the corresponding port number leading to that neighbor in $\alpha(w).\nxt$. 
\item If all neighbors have a settler (no neighbor fully unsettled), $\alpha(w).\nxt=\bot$ is true.
\item \texttt{Async\_Probe()} will finish in $O(\log k)$ epochs.
\item \texttt{Guest\_See\_Off()} will settle the collected helper agents back to their homes.
\item \texttt{Guest\_See\_Off()} will finish in $O(\log k)$ epochs.
\end{itemize}

\noindent{\bf Forward and backtrack moves.}
In the beginning of $\rootasync$, 
all agents in $\A$ are located at node $s$. The agent with the smallest ID $a_{\min}(s)$ settles at node $s$ and sets $a_{\min}(s).\parent\leftarrow \bot$. We denote $a_{\min}(s)$ as $\alpha(s)$, the agent settled at node $s$.   Then, as long as there is (at least) a fully unsettled node in $N(s)$, all explorers move to that node together (which is a forward move). After each forward move from a node $w$ to $u$, the agent with the smallest ID among $A(u)$ settles on $u$, and $\alpha(u).\parent$ is set to the port of $u$ leading to $w$.
When the current node $u$  has no fully unsettled node in $N(u)$, all explorers in $A(u)$ move to the parent of $u$ (which is a  backtrack move).  Finally, $\amax$ terminates when it settles (since $\amax$ settles last) and dispersion is achieved. Since the number of agents is $k\leq n$, the DFS stops after $\amax$ makes a forward move $k-1$ times. The agent $\amax$ makes a backward move at most once from every visited node, i.e., total at most $k-1$ times. Therefore, excluding  \texttt{Async\_Probe()}, the execution of $\rootasync$ completes in $O(k)$ epochs. Notice that  \texttt{Async\_Probe()} is invoked at most $2(k-1)$ times, once after each forward move and once after each backward move. 
Since a single invocation of \texttt{Async\_Probe()} requires $O(\log k)$ epochs, the time complexity of $\rootasync$ becomes $O(k \log k)$ epochs.

Regarding memory, an agent handles several $O(\log \Delta)$-bit variables, $\nxt$, $\checked$, $\parent$, as well as the port number that the settler $\alpha(u)$ needs to remember in order to return to node $u$ from node $w$ during \texttt{Guest\_See\_Off()}. Every other variable can be stored in a constant space. 
An agent needs $\lceil \log k\rceil $ bits to store its ID.  
Therefore, the memory complexity is $O(\log (k+\Delta))$ bits.

\begin{theorem}
\label{theorem:rootal}
{\it Algorithm $\rootasync$  solves dispersion within $O(k\log k)$ epochs using $O(\log (k+\Delta))$ bits per agent.}
\end{theorem}

\begin{algorithm}[t]
\caption{$\rootasync$}\label{alg: async_rooted_memory}
\SetKwInOut{Input}{Input}
\SetKwInOut{Output}{Output}
$a_{\min}.settled \gets \top$, 
$a_{\min}.parent \gets \bot$\;
$\alpha(s).\mathtt{treelabel} \leftarrow 0$\;
\While{$a_{\texttt{max}}.settled = \bot$}
{
    
    $w \gets$ the node $\alpha(a_{\max})$ where $a_{\max}$ is positioned\;
    \texttt{Async\_Probe()} at $w$\;
    \eIf{$\alpha(w).next \neq \bot$}
    {
        $u \gets N(w, \alpha(w).next)$\; 
        agents in $A(w)$ go to $u$\;
        $a_{min}\leftarrow$ the smallest ID agent in $A(u)$\;
        $a_{\min}.settled \gets \top$,
        $a_{\min}.parent \gets a_{\max}.pin$\;
        $\alpha(u).\mathtt{treelabel} \leftarrow a_{\max}.treelabel$\;
    }
    {
        agents in $A(w)$ exit $w$ via $\alpha(w).parent$\;
    }
}
\end{algorithm}

\section{General Dispersion Algorithms}
\label{section:general_dispersion}
In this section, we discuss the dispersion of agents from general initial configurations, i.e., $k\leq n$ agents are initially on multiple nodes. 
Suppose agents are initially on $\ell$ nodes. The $\ell$ DFSs start in parallel from $\ell$ nodes.
We extend the idea of merging from Kshemkalyani and Sharma~\cite{KshemkalyaniS21-OPODIS} [OPODIS'21] to deal with merging the DFSs if one DFS meets another. We refer in this paper to this merging algorithm as the KS algorithm.

In the KS algorithm, multiple DFSs grow concurrently. The size of a DFS $i$ at any point in time is the number of settled agents $d_i$ in it. When one DFS meets another, the smaller sized DFS collapses and gets subsumed in the larger sized DFS. Subsumption and collapse essentially mean that the agents settled from the subsumed DFS are collected and given to the subsuming DFS to extend its traversal. This is done via a common module to re-traverse an already identified DFS component with nodes having the same $treelabel$. Such re-traversal occurs in procedures {\em Exploration}, {\em Collapse\_Into\_Child}, and {\em Collapse\_Into\_Parent}, and can be executed in $4d_i$ steps. After such a subsumption, the subsuming DFS continues growing from where it last left off, this time with added agents from the subsumed DFS(s). Operation in these two phases -- growing and subsuming -- continues until all agents are settled and there are no more meetings between DFSs. We term the DFS that subsumes other met DFS(s) at each meeting as the winner DFS. Thus in the KS algorithm, the winner DFS alternates between the two phases -- growing, and subsuming. The sum of all subsuming times is $O(k)$ time\footnote{In~\cite{KshemkalyaniS21-OPODIS}, the subsuming time assumed a loose bound of $O(\min\{m,k\Delta\}$) because that did not affect the overall time complexity of the algorithm due to the larger time of $O(\min\{m,k\Delta\})$ of the growing phase. However, the sum of all subsuming times is $O(k)$ as each re-traversal of DFS component $i$ can be done in $O(d_i)$ steps.}, whereas the sum of all growing times is $O(\min\{m,k\Delta\})$. 

\subsection{\texorpdfstring{\sync}{} General  Dispersion Algorithm}
\label{sec:gensync}
Starting from an arbitrary initial configuration, multiple DFSs grow concurrently as per Algorithm $\rootsync$ (Section~\ref{section:RootedSync}). Essentially the KS algorithm is run but in the growing phase Algorithm $\rootsync$ is followed. When there is a meeting between two DFS trees, the subsuming phase of the KS algorithm is run. The re-traversal of DFS $i$ in the subsuming phase has to be adapted to account for the oscillating agents -- specifically, when visiting a node in the re-traversal, the agent waits there for up to 6 rounds so that if the settled agent is an oscillating agent, the oscillating agent visits its home within that period -- and this introduces an $O(1)$ factor in time complexity, thereby the overall re-traversal still takes $O(d_i)$ time. 
The sum of all the subsuming phases is thus $O(k)$ and, from the analysis of Algorithm $\rootsync$, the sum of all the growing phases is $O(k)$. Thus, the general initial configuration can be solved in $O(k)$ rounds in {\sync}. This leads to the following result.

\begin{theorem}
\label{theorem:gensync}
{\it Starting from general initial configurations,  dispersion can be solved in $O(k)$ rounds using $O(\log (k+\Delta))$ bits per agent in {\sync}.
}
\end{theorem}

\subsection{\texorpdfstring{\async}{} General  Dispersion Algorithm}
\label{sec:genasync}
The KS algorithm is run, but in the growing phases, we use the Algorithm $\rootasync$ (Section~\ref{section:RootedAsync}).
As in the {\sync} general algorithm, the growing phases 
(of Algorithm $\rootasync$) and the subsuming phases in the KS algorithm need to be modified to deal with an empty node which is the home node of a settled agent that is currently helping with probing. In particular, how to detect that DFS $i$ meets DFS $j$ (and deal with it) when the ``meeting node'' does not have an agent from DFS $j$ because that agent is helping with probing. The modifications required are detailed below.

In Lines 10-11 of procedure {\tt Async\_Probe()}, if $a_i$ finds that the settler at $u_i$, $\alpha(u_i)$ belongs to a different DFS (determined by having a different $treelabel$ $t'$ than that of $a_i$ which is $t$), then the DFS $t$ meets DFS $t'$; in this case the procedures {\em Exploration}, {\em Collapse\_Into\_Child} and {\em Collapse\_Into\_Parent} for DFS $t$ are invoked as in the KS algorithm after executing {\tt Guest\_See\_Off()} for all the agents that were collected during {\tt Async\_Probe()}. (If there are multiple such trees $t'$, the $treelabel\; t'$ settler at the $u_i$ reachable by the lowest numbered port is considered.) 
With this change, the executions in the growing phases occur in $O(k \log k)$ epochs, and as the executions in the subsuming phases in the KS algorithm also occur in $O(k)$ epochs, the overall run-time is $O(k\log k)$.

Algorithm~\ref{alg: async_rooted_memory} follows the same steps in conjunction with the KS algorithm for the general asynchronous case but with additional modifications. 
If an agent $a$ with $treelabel$ $t$ is helping (in parallel) {\tt Async\_Probe()} procedure to search for the next DFS node of DFS $t$, then the agent $a$ leaves its home node $u$ to do probing. During that time, $u$ might be acquired by an agent $a'$  with $treelabel$ $t'$ in the absence of $a$. We call this scenario {\em squatting} due to agent $a'$.  
If DFS $t'$ was terminated with the settling of $a'$, then when $a$ returns to $u$, DFS $t'$ is said to meet DFS $t$ and the KS algorithm executes the subsuming steps after the meeting as usual. If DFS $t'$ was not over with the settling of $a'$, then DFS $treelabel$ $t'$ unaware of the squatting act $a'$, does the DFS tree traversal as per protocol. In doing so, observe that it will follow the path of DFS $t$ either downstream away from the root, or upstream towards the root (and then possibly downstream) because agents perform DFS based on port numbering.
This is because all DFSs ($t$ and $t'$) perform their DFSs based on port numbering beginning with port 1,2, $\ldots$, $\delta_i$. 
This approach requires the following assumption for the system model as follows.  For any edge $(u,v)$, the two ports cannot be labelled $(1,1), (1,2),(2,1)$, or $(2,2)$, subject to the following exceptions. (i) Port number 1 is permitted if that is the only port at a node. (ii) Port number 2 is permitted if there are only two ports at a node. This assumption is required to ensure that when the path of DFS $t'$ meets the path of DFS $t$ (at a node which is currently vacant due to the settled node of $t$ helping in {\tt Async\_Probe()}), DFS $t'$ proceeds either in the upstream or downstream direction of DFS $t$, 
irrespective of the port numbers from which the two DFSs entered that node.

There might be the case that agents of DFS $t$ on that path followed by the agents of $t'$ are also helping in the probing and are not present at their home nodes.  Eventually, the head of $t'$ either (i) finds an agent of $t''$ (may or may not equal $t$) or (ii) all the agents get settled along the path. For case (ii), the last agent of $t'$ to settle (having no child) creates a meeting with $t$ when the corresponding agent from $t$ returns, and the KS algorithm executes the subsuming steps. For case (i), a meeting occurs with $t''$, and the subsequent steps of KS subsumption are executed. (This can happen recursively where $t''$ is distinct from $t$ but the recursion terminates when $t''$ equals $t$ or case (ii) holds.) If in the KS subsumption steps, a node is reached where no agent 
is present when it should be whereas an agent from another DFS is settled there (that other agent must be squatting), the agents wait there; the missing agent is helping in {\tt Async\_Probe()} and is guaranteed to return there in $O(\log k)$ epochs.

Note that when agent $a$ of DFS $t$ returns to its home $u$ where $a'$ of DFS $t'$ is squatting, both can continue to co-exist there for a while (until one gets subsumed by or subsumes another DFS due to a meeting at possibly another node; 
see cases (i) and (ii) in the above paragraph.). The node $u$ has an overlay of multiple DFSs and acts as multiple independent virtual nodes where the different DFSs' agents have settled. A third (or more) agent $a''$ from DFS $t''$ can also be at $u$ and when it arrives, $a''$ is said to meet $a'$ or $a$, it does not matter which. The KS algorithm steps for subsumption are executed at the meeting. 

Two further modifications are required to the KS algorithm. 
\begin{itemize}
    \item [1.] When a DFS $t'$ is collapsing, if a node has agents from $t'$ and other DFSs (such as $t$), only the tree $t'$ collapses, and only agents of $t'$ are collected/participate in the collapse. 
    \item [2.] In the KS algorithm, $d_i$ denotes the number of settled agents in DFS $i$, and $parent(i)$ denotes the DFS that DFS $i$ has met. $head(i)$ denotes the last node in DFS $i$ at which agents of DFS $i$ are currently present. In {\em Parent\_Is\_Collapsing} (Algorithm 3 of KS, lines 28-29), we require the change given below following the boldface ``{\bf after}''.\\ 
    ``If $(d_i > d_{parent(i)})$ and $head(i)$ junction is not locked and remains unlocked until $parent(i)$'s collapse reaches $head(i)$, then unsettled robots get absorbed in $parent(i)$ during its collapse {\bf after} the agent $a''$ of DFS $i$ at $head(i)$ informs its parent node $w$, i.e., $parent(a'')$ in DFS $i$, to create a meeting with $t$, so that (if and when if not already) the corresponding agent from $t$ returns to $w$ after helping with parallel probing, the KS algorithm executes the subsequent subsuming steps.'' 
\end{itemize}

The sum of all the subsuming phases of KS is $O(k)$ and, from the analysis of Algorithm $\rootasync$, the sum of all the growing phases is $O(k\log k)$. All the above changes do not add asymptotically at this time. 
The waiting for the agent of DFS $t$ to return to its home node after helping with parallel probing increases the subsuming time of $O(k)$ by a factor of at most $O(\log k)$ as the wait time of 
$O(\log k)$ can be incurred $O(k)$ times --- as $O(k)$ number of times {\tt Async\_Probe()} is invoked when nodes may be at the head of some DFS in the growing phase. 
Thus, the general initial configuration can be solved in $O(k\log k)$ epochs in {\async}. This gives the following result.

\begin{theorem}
\label{theorem:genasync}
{\it Starting from general initial configurations,  dispersion can be solved in  $O(k\log k)$ epochs using $O(\log(k+\Delta))$ bits per agent in {\async}.
}
\end{theorem}

\section{Concluding Remarks}\label{sec: conclusion}
In this paper, we have developed two novel techniques, one for {\sync} and another for {\async}. For {\sync}, we have introduced a technique of leaving some DFS tree nodes empty during the DFS traversal which are covered through oscillation. This technique allowed to solve dispersion with optimal time complexity $O(k)$ in {\sync} with memory only $O(\log (k+\Delta))$ bits per agent, improving the best previously known time bound by a factor of $O(\log^2k)$.  For {\async}, we developed a technique to extend the previous idea in {\sync}, solving dispersion in $O(k\log k)$ time complexity  with memory only $O(\log (k+\Delta))$ bits per agent. 
This is almost a quadratic improvement compared to the best previously known time bound of $O(\min\{m,k\Delta\})$ in {\async} with $O(\log (k+\Delta))$ bits per agent. Closing the $O(\log k)$ factor gap in time complexity in {\async} remains an intriguing open question for future work. 
The starting point toward addressing this open question to see whether our {\sync} technique can be extended to {\async} with no overhead on time complexity.

\bibliographystyle{plain}
\bibliography{reference}

\end{document}